\documentclass[aps,physrev,preprint,amsmath,amssymb,]{revtex4-2}
\usepackage{graphicx}
\usepackage{braket}
\usepackage{dcolumn}% Align table columns on decimal point
\usepackage{amsthm}
\usepackage{enumitem,kantlipsum}

\newtheorem{theorem}{Theorem}
\newtheorem{corollary}{Corollary}[theorem]

\newtheorem{principle}{Principle}

\begin{document}

%\preprint{}

\title{Quantum Resource Complementarity in Finite-Dimensional Systems}


\author{Justin K. Edmondson}
\email{Contact author: justin.k.edmondson.civ@us.navy.mil}
%\thanks{}
\affiliation{Farragut Technical Analysis Center, USN, Washington DC, USA.}


\date{\today}

\begin{abstract}
Quantum resources such as entanglement, information redundancy, and coherence enable revolutionary advantages but obey fundamental tradeoffs. We present a unified geometric constraint governing three core operational tasks: teleportation (\(q_1\)), cloning (\(q_2\)), and coherence-based metrology (\(q_3\)). For any tripartite quantum state \(\rho_{ABC}\), we show the tight inequality \(q_1^2 + q_2^2 + q_3^2 \leq 1\) confines all physically achievable resources to the positive octant of the unit ball. This Quantum Information Resource Constraint (QIRC) reflects an exclusion principle intrinsic to Hilbert space: optimizing one task necessitates sacrificing others. Crucially, \(q_1, q_2, q_3\) are experimentally measurable, making QIRC falsifiable in quantum platforms. Unlike abstract quantum resource theories (QRT) that quantify resources through entropy or monotones, our framework is fundamentally operational, deriving tight constraints from measurable task fidelities in teleportation, cloning, and metrology. The emergent \(\ell^2\)-norm exclusion is irreducible to existing QRT axioms. Remarkably, we demonstrate the resource norm \(\mathcal{I} = q_1^2 + q_2^2 + q_3^2\) is conserved under symmetry-preserving unitaries (quantum resource covariance principle) but contracts irreversibly under decoherence. This work establishes a fundamental link between quantum information geometry, symmetry, and thermodynamics.  
\end{abstract}

%\keywords{}

\maketitle

\section{Introduction\label{S1}}

Quantum mechanics grants physical systems remarkable capabilities such as entanglement enabled teleportation, high-precision metrology, and extraordinary computation performance. Yet these advantages are not without fundamental limits imposed by the very structure of quantum theory. Information cannot be perfectly copied (no-cloning theorems), entanglement cannot be freely shared (monogamy), and maintaining coherence comes at a cost. While these individual limitations are well understood, in this paper we ask a different question: how do these quantum resources compete when a quantum system must simultaneously support multiple operational tasks?

In this work, we introduce a unified geometric framework that governs the coexistence and distribution of quantum advantages. We derive the Quantum Information Resource Constraint (QIRC) theorem establishing that no finite-dimensional quantum system can simultaneously maximize its capacity for teleportation, cloning, and coherence-based metrology. Specifically, for or any tripartite quantum state \(\rho_{ABC}\), we define three operationally meaningful fidelities, teleportation fidelity (\(q_1\)), cloning fidelity (\(q_2\)), and coherence-task fidelity (\(q_3\)), each normalized to the interval \([0,1]\). These fidelities form a resource vector \((q_1, q_2, q_3)\) of all physically realizable configurations, which is confined to the positive octant of a unit ball in \(\mathbb{R}^3\) by the inequality 
\[
q_1^2 + q_2^2 + q_3^2 \leq 1.
\]  
This constraint does not emerge from any specific protocol. Rather, it reflects an intrinsic tradeoff imposed by the geometry of Hilbert space itself. Just as quantum mechanics enables advantages like entanglement and superposition, it also imposes exclusionary relationships among them. In this sense, the QIRC theorem serves as a generalized no-go principle, revealing that a quantum system’s ability to excel in one task inherently limits its performance in others.  

What sets this result apart is its operational and geometric unification. Operationally, each fidelity \(q_i\) corresponds to a well-defined quantum protocol. \(q_1\) teleportation capacity, \(q_2\) cloning performance, and \(q_3\) metrological utility or phase estimation. Each fidelity is measurable making the QIRC constraint experimentally testable/falsifiable in platforms ranging from superconducting qubits to photonic networks. Geometrically, the \(\ell^2\)-norm bound generalizes known limits, such as the no-cloning theorem and monogamy of entanglement, into a single, unifying structure. When one fidelity approaches its maximum, the others must vanish, creating a quantum analogue of a Pareto frontier where optimizing one resource comes at the expense of the others.

While quantum resource theories (QRT) \cite{Nielsen_Chuang_2010,Brandaoetal2015,ChitambarGour2019} formalize individual limitations (e.g., entanglement distillation), they fail to capture the geometric competition between teleportation, cloning, and coherence tasks. Our exclusion principle arises not from abstract monotonicity, but from the Hilbert space structure itself, quantified through experimentally falsifiable fidelities.

The QIRC framework extends dynamically. We demonstrate that the total resource norm \(\mathcal{I} = q_1^2 + q_2^2 + q_3^2\) is conserved under symmetry-preserving unitary evolution, much like energy in closed classical systems. Conversely, under decoherence or noise, \(\mathcal{I}\) contracts irreversibly, acting as a quantum Lyapunov function that quantifies the degradation of operational capacity. This dynamical perspective bridges quantum information theory with thermodynamics, offering new insights into how resources dissipate in open systems.  

The immediate implications of this work span theory and experiment. The QIRC provides a falsifiable benchmark for quantum devices, guiding the design of protocols that navigate fundamental tradeoffs. The \(\ell^2\)-norm \(||(q_1,q_2,q_3)||_{\ell^2}^2 \leq 1\) defines achievable regions for quantum benchmarking multi-task optimization. Tradeoffs guide error correction resource allocation in fault-tolerant designs. The constraint may inform thermodynamic limits via resource degradation. It also opens speculative connections to foundational questions, such as whether informational constraints might influence emergent spacetime structure (a possibility we explore cautiously).  

The paper is structured as follows. In Section \ref{S2}, we rigorously defining the fidelity coordinates \(q_1, q_2, q_3\) and their operational significance. Section \ref{S3} establishes the QIRC theorem and characterizes the feasible region of quantum resource configurations. Section \ref{S4} presents key corollaries, including the tightness of the bound, the convexity of the feasible set, and the monotonic degradation of resources under noise. In Section \ref{S5}, we suggest a Quantum Resource Complementarity (QRC) Principle, which governs the conservation and flow of resources under unitary and noisy dynamics. Section \ref{S6} discusses experimental falsifiability and broader theoretical implications. Section \ref{S7} concludes with a synthesis of the framework’s significance and future directions.  

By unifying kinematic constraints, dynamical symmetry, and thermodynamic irreversibility, the QIRC framework redefines quantum advantages not as isolated phenomena but as interdependent facets of a deeper geometric structure. This perspective offers new quantum informational tools for both theoretical exploration and practical innovation.

\section{Operational Resource Functionals\label{S2}}

We consider a tripartite quantum system described by a density matrix \(\rho_{ABC}\) defined on a finite-dimensional Hilbert space \(\mathcal{H}_A\otimes\mathcal{H}_B\otimes\mathcal{H}_C\) with local dimensions \(d_A = \dim\mathcal{H}_A\), and similarly for \(B\) and \(C\). In this section, we define three operationally accessible, normalized quantities \(q_1\), \(q_2\), and \(q_3\), each formulated as a functional of the global state \(\rho_{ABC}\), and built from standard quantum information protocols.

We define \(q_1\) as a normalized \textit{teleportation fidelity advantage coordinate}
\begin{equation}\label{q1}
    q_1(\rho_{AB}) \equiv \frac{F_{\text{tele}}(\rho_{AB}) - F_{\text{classical}}}{F_{\text{quantum}} - F_{\text{classical}}} = 3 F_{\text{tele}}(\rho_{AB}) - 2
\end{equation}
since \(F_{\text{classical}} = 2/3\) and \(F_{\text{quantum}} = 1\) (see Appendix \ref{A13}) \cite{Bennettetal1993,Popescu1995}. \(q_1\) quantifies how well subsystem \(A\) can be used to teleport quantum information to subsystem \(B\) using the shared bipartite state \(\rho_{AB}\) as a resource under Local Operations and Classical Communication (LOCC; see Appendix \ref{A13}). It maps the (average) teleportation fidelity \(F_{\text{tele}}\) onto the unit interval \(q_1 \in [0,1]\) and reflects the fraction of quantum advantage in teleportation over the best classical strategy. 

We define \(q_2\) as a normalized \textit{cloning fidelity coordinate}
\begin{equation}\label{q2}
    q_2(\rho_{AC}) \equiv \frac{F_{\text{clone}}(A\to C)}{F_{\text{perfect}}} = F_{\text{clone}}(A\to C)
\end{equation}
where the normalization \(F_{\text{perfect}} = 1\) corresponds to ideal cloning fidelity (see Appendix \ref{A13}) \cite{Wootters1982,BuczekHillery1996}. \(q_2 \in [0,1]\) characterizes the ability to duplicate quantum information held in subsystem \(A\) to another subsystem \(C\), while \(A\) is already correlated with \(B\); i.e., \(q_2\) decrease as \(A\)'s entanglement with \(B\) increases. This tests the shareability of information from \(A\) across \(B\) and \(C\), and is constrained by the monogamy of entanglement and the no-broadcasting theorem \cite{Coffmanetal2000,Barnumetal1996}.

Finally, we define \(q_3\)  as a normalized coherence fidelity coordinate, quantifying the ability of subsystem \(A\) to support coherence-based quantum tasks such as phase estimation, quantum metrology, or interference experiments. Specifically, we define
\begin{equation}\label{q3}
    q_3(\rho_A) \equiv \frac{F_Q(\rho_A, H)}{F_Q^{\max}(d_A)}
\end{equation}
where \(F_Q(\rho_A, H)\) denotes the quantum Fisher information of the state \(\rho_A\) with respect to a fixed Hermitian generator \(H\) (often taken as a Hamiltonian or observable defining a preferred coherence basis), and \(F_Q^{\max}(d_A)\) is the maximal Fisher information attainable for a pure state in a Hilbert space of dimension \(d_A\) under generator \(H\) (see Appendix \ref{A13}) \cite{BraunsteinCaves1994,Giovannettietal2011,Petz1996}. This normalization ensures \(q_3 \in [0,1]\) for all finite-dimensional systems.

Each coordinate \(q_a\) is a well-defined functional of the global state \(\rho\), derived through reduced and joint states \(\rho_A\), \(\rho_{AB}\), and \(\rho_{AC}\), and defined by standard quantum information tasks, i.e., LOCC teleportation \(q_1\) , approximate cloning \(q_2\) , and coherence-based quantum task fidelity \(q_3\) . Collectively, the map
\begin{equation}\label{InfoExtraction}
    \rho_{ABC} \mapsto \left( q_1(\rho_{AB}), q_2(\rho_{AC}), q_3(\rho_A) \right)
\end{equation}
defines an "information resource profile" of the subsystem \(A\) that captures the operational capacities of subsystem \(A\), i.e., quantum communication with \(B\), duplication to \(C\), and local coherence retention. The coordinates are dimensionless and bounded \((q_1, q_2, q_3) \in [ 0, 1 ]^3\), monotonic under local noise (CPTP maps) \cite{Nielsen_Chuang_2010}, and operationally measurable through task fidelities.

All three coordinates \(q_1, q_2, q_3\) are now defined as normalized task fidelities derived from operational quantum channels acting on \(\rho_{AB}\), \(\rho_{AC}\), and \(\rho_A\) respectively. While the coherence fidelity \(q_3\)  remains single-subsystem in scope, its task-based definition places it on equal theoretical footing with the teleportation and cloning fidelities.

\section{Information Resource Tradeoff Constraint\label{S3}}

We now state and prove the central result of this paper. A universal constraint on the joint extractability of the quantum information resources entanglement (though teleportation fidelity), shareability (through cloning fidelity), and local coherence (through fidelity in canonical coherence-based tasks) from any quantum subsystem. This result encodes a fundamental tradeoff surface in quantum information geometry.

\begin{theorem}[Quantum Information Resource Constraint (QIRC)]\label{T1}
For any tripartite quantum state \(\rho_{ABC}\) on \(\mathcal{H}_A \otimes \mathcal{H}_B \otimes \mathcal{H}_C\), the operational quantities \(q_1\) (teleportation fidelity advantage), \(q_2\) (cloning fidelity), and \(q_3\) (coherence-task fidelity) satisfy: 
\begin{equation}\label{InfoConstraintTM}
   q_1^2 + q_2^2 + q_3^2 \leq 1
\end{equation}
defining a convex tradeoff region \(Q = B^3 \cap [0,1]^3 \subset \mathbb{R}^3\).
\end{theorem}

Before proving Theorem \ref{T1}, we provide a physical interpretation and outline of the argument. The theorem expresses a fundamental constraint on the simultaneous attainability of quantum informational resources. A subsystem cannot be maximally entangled with one party (enabling perfect teleportation) \cite{Bennettetal1993}, perfectly cloned or shared with another (violating the no-broadcasting theorem if unrestricted) \cite{Wootters1982,Barnumetal1996}, and simultaneously maintain full quantum coherence \cite{Colesetal2017,Berta2010}. The three resource quantities \((q_1,q_2,q_3)\) correspond to teleportation fidelity, cloning shareability, and coherence-task fidelity, respectively. Together they define a vector in a quantum information resource space, with Theorem \ref{T1} asserting that all physically realizable configurations must lie within a convex feasibility region (see figure \ref{FIGInfoRegion}); a tradeoff surface expressing the fundamental incompatibility of jointly maximizing these capabilities.

\begin{figure}[b]
\centering
\includegraphics[width=0.8\textwidth]{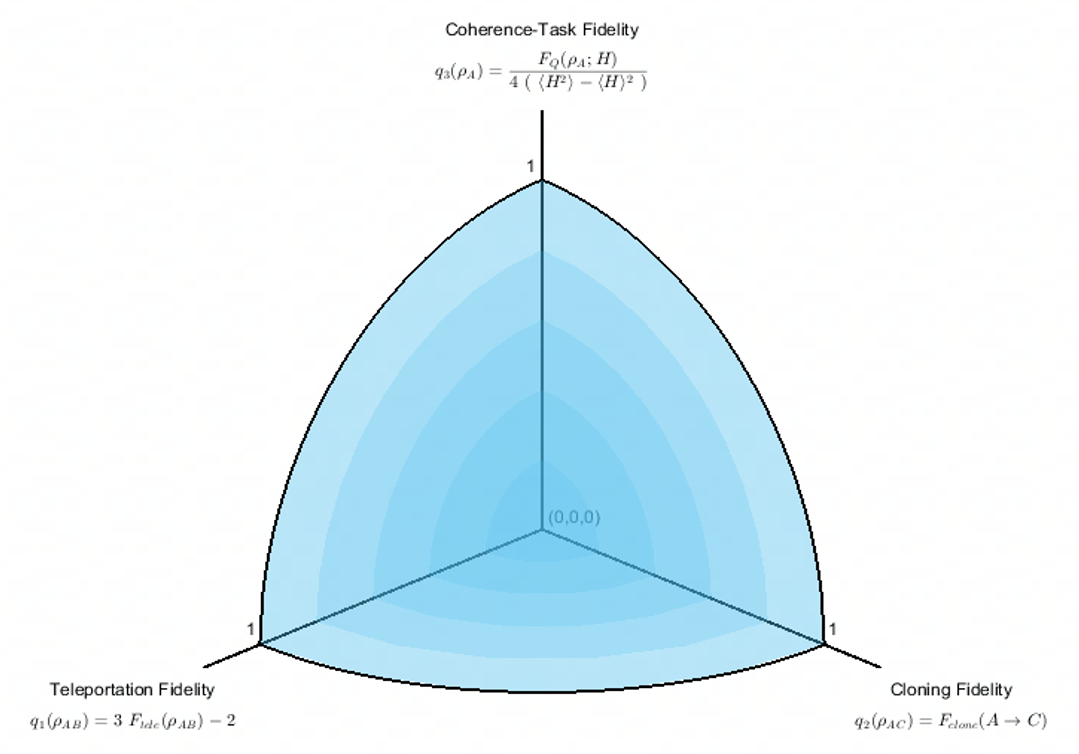}%
\caption{Quantum information resource feasibility region \(Q = B^3 \cap [0,1]^3 \subset \mathbb{R}^3\).\label{FIGInfoRegion}}
\end{figure}

We outline the structure of the proof in three parts: 1. boundedness and monotonicity; 2. mutual incompatibility; and 3. entropic tradeoff.

First, each quantity \(q_a \in [0,1]\) is normalized to reflect a task-specific operational performance. These quantities are monotonic under noise and local operations. Decoherence reduces coherence fidelity \(q_3\), entanglement cannot be increased through LOCC which limits \(q_1\) \cite{Nielsen_Chuang_2010}, and cloning fidelity is constrained by no-broadcasting limits which bounds \(q_2\) \cite{Barnumetal1996}. Extremal values such as \(q_1 \to 1\) imply specific structural constraints, e.g., maximal entanglement with \(B\), maximal purity for coherence, or optimal cloning performance.

Second, the extremal values are pairwise exclusive. Maximal entanglement with \(B\) (\(q_1 \to 1\)) forces \(\rho_A\) to be maximally mixed \cite{Horodecki_2009}, thereby prohibiting strong correlation with \(C\) (\(q_2 \to 0\)), and and destroying local coherence fidelity (\(q_3 \to 0\)) \cite{Giovannettietal2011,TothApellaniz2014}. Similarly, high coherence fidelity (\(q_3 \to 1\)) requires purity of \(\rho_A\), which limits possible correlations with both \(B\) and \(C\) (\(q_1, q_2 \to 0\)) \cite{Colesetal2017,Nielsen_Chuang_2010}. Strong cloning (\(q_2 \to 1\)) implies significant correlation with \(C\) (\(q_1 \to 0\)), again implying a mixed \(\rho_A\) and loss of coherence fidelity (\(q_3 \to 0\)).

Third, Interpreting \(\rho_{ABC}\) as a purification of \(\rho_A\), we have \(S(\rho_A) = S(\rho_{BC})\). This constraint implies that high correlations with both \(B\) and \(C\) (i.e., large \(q_1, q_2\)) require high entropy in \(\rho_A\), while high coherence \(q_3\) requires low entropy. This entropic tension underpins the convexity of the feasible region and the exclusion of certain corners of the resource space \cite{Colesetal2017,Berta2010,Petz2011}.

Together, these arguments constrain the 3-tuple \((q_1,q_2,q_3)\)  to a bounded convex region in \(\mathbb{R}^3\), i.e., the positive octant of the unit 3-ball \(Q \subset \mathbb{R}^3\). Extremal points along each axis are realizable, but no physical quantum state permits simultaneous saturation of more than one resource. Theorem \ref{T1} formalizes this limitation and delineates the boundary of the informational tradeoff space. 

We now present the proof.
\begin{proof} (By contradiction.)

Let \(\rho_{ABC}\) be a tripartite quantum state with reduced states \(\rho_{AB} = \text{Tr}_C(\rho_{ABC})\), \(\rho_{AC} = \text{Tr}_B(\rho_{ABC})\), and \(\rho_A = \text{Tr}_{BC}(\rho_{ABC})\). The operational quantities are defined as:
\begin{enumerate}[label=(\alph*), itemsep=0.em, wide, align=left]
    \item Teleportation fidelity advantage. \(q_1(\rho_{AB}) = 3F_{\text{tele}}(\rho_{AB}) - 2 \in [0,1]\), where \(F_{\text{tele}}\) is the average fidelity for teleporting states from \(A\) to \(B\) using \(\rho_{AB}\);
    \item Cloning fidelity. \(q_2(\rho_{AC}) = F_{\text{clone}}(A \to C) \in [0,1]\), normalized to the optimal cloning fidelity; and
    \item Coherence-task fidelity. \(q_3(\rho_A) = F_Q(\rho_A, H)/F_Q^{\max}(d_A) \in [0,1]\), where \(F_Q\) is the quantum Fisher information for a fixed observable \(H\).
\end{enumerate}

Assume there exists a state \(\rho_{ABC}\) violating the theorem
\[
q_1^2 + q_2^2 + q_3^2 > 1
\]
We analyze the implications when any one resource approaches its maximum (\(q_a \to 1\)).
\begin{enumerate}[wide, align=left]%[label=(\alph*), itemsep=0.em, wide, align=left]
    \item Maximal teleportation fidelity: \(q_1 \to 1\) implies \(\rho_{AB}\) approaches a maximally entangled state (e.g., \(\ket{\Phi^+}_{AB}\); see Appendix \ref{A11}). Then, \(\rho_A = \text{Tr}_B(\ket{\Phi^+}\bra{\Phi^+}) = I/d_A\) is maximally mixed. By the no-broadcasting theorem, maximal entanglement with \(B\) prohibits cloning to \(C\) (\(q_2 \to 0\)). And there is no coherence in any basis \(q_3 \propto F_Q(\rho_A, H) = 0\) for \(\rho_A = I/d_A\). Thus, \(q_1 \to 1 \implies q_2^2 + q_3^2 \to 0\), violating the assumption \(q_1^2 + q_2^2 + q_3^2 > 1\).
    \item Maximal cloning fidelity: \(q_2 \to 1\) implies \(\rho_{AC}\) permits perfect cloning from \(A\) to \(C\). Consequently, the no-broadcasting theorem requires \(\rho_{AB}\) and \(\rho_{AC}\) to be classically correlated, so \(q_1 \to 0\) (i.e., no entanglement). Moreover, cloning degrades purity; that is, for \(q_2 \to 1\), then \(\rho_A\) must be mixed, so \(q_3 \to 0\). Thus, \(q_2 \to 1 \implies q_1^2 + q_3^2 \to 0\), again violating the assumption. We remark, perfect cloning is unachievable for finite-dimensional systems (e.g., optimal qubit cloning has \(F_{\text{clone}} = 5/6\) \cite{BuczekHillery1996}), but the bound holds asymptotically.
    \item Maximal coherence fidelity: \(q_3 \to 1\) implies \(\rho_A\) is pure and maximizes \(F_Q(\rho_A, H)\) (e.g., \(\rho_A = |+\rangle \langle +|\) for \(H = \sigma_x\)). Then, purity implies factorization of the state \(\rho_{ABC} = \rho_A \otimes \rho_{BC}\), thereby eliminating correlations \(q_1 = 0\) (i.e., no entanglement with \(B\)), and \(q_2 = 0\) (no cloning to \(C\)). Thus, \(q_3 \to 1\) implies \(q_1^2 + q_2^2 \to 0\), contradicting the assumption.
\end{enumerate}

The extremal cases above saturate the bound \(q_1^2 + q_2^2 + q_3^2 = 1\) when only one resource is active. For mixed states or intermediate resource values, each \(q_i\) is non-increasing under CPTP maps (Corollary \ref{C3} monotonicity). Moreover, the set of achievable \((q_1, q_2, q_3)\) is convex, as mixtures of states cannot exceed the norm of their extremal points (Corollary \ref{C2} convexity). Thus, the inequality holds for all states 
\[
q_1^2 + q_2^2 + q_3^2 \leq 1.
\]

To demonstrate entropic consistency, observe for pure \(\rho_{ABC}\), the entropy of \(\rho_A\) bounds the mutual information with \(B\) and \(C\)
\[
I(A:B) + I(A:C) \leq 2S(\rho_A).
\] 
(see Appendix \ref{A2}). Since \(q_1^2 + q_2^2 \lesssim \left(2S(\rho_A)/\log d_A\right)^2\) and \(q_3^2 \lesssim 1 - S(\rho_A)/\log d_A\), the \(\ell^2\)-norm constraint emerges naturally from entropic tradeoffs.

Finally, althought the coherence measure \(q_3\) is defined relative to a fixed observable \(H\), the bound \(q_1^2 + q_2^2 + q_3^2 \leq 1\) holds for any choice of \(H\). Thus, the spherical symmetry reflects the incompatibility of resources across all possible tasks.

Hence, the assumption \(q_1^2 + q_2^2 + q_3^2 > 1\) leads to contradictions in all extremal cases, and convexity extends the result to all states. Thus, the theorem holds. 
\end{proof}

We have established the inequality (\ref{InfoConstraintTM}) as a universal bound on the extractable informational capacities of quantum subsystems. No state satisfying all constraints violates the bound. Therefore, the inequality is tight, convex, and fundamental. This result reflects a multi-task exclusion principle governed not by abstract entropy measures alone, but by directly testable operational quantities, including task-based coherence fidelity. The geometry of the information resource tradeoff region \(Q \subset \mathbb{R}^3\) emerges from a deep structure of interference between entanglement, coherence, and shareability. This spherical constraint \(||(q_1,q_2,q_3)||_2^2 \leq 1\) reflects a fundamental property of quantum information geometry that cannot be recovered from entropic QRTs (e.g., \cite{Colesetal2017,Brandaoetal2015}) due to its task-theoretic origin.

The above proof establishes the exclusion constraint \(||( q_1,q_2,q_3)||_{\ell^2}^2\) by contradiction; that is, assuming simultaneous saturation of multiple informational coordinates leads to conflicts with known structural constraints in quantum theory. While sufficient in conjunction with convexity (Corollary \ref{C2}), this argument does not directly derive the inequality for general quantum states. An alternative constructive derivation based on task-based fidelity monotonicity could provide deeper insight into the geometry of the feasible region \(Q \subset \mathbb{R}^3\). We outline such an approach in Appendix \ref{A2} and leave full development to future work.

\section{Information Resource Space Geometry\label{S4}}

Theorem \ref{T1} establishes that the accessible information resources and capabilities of a quantum subsystem are constrained to a region \(Q \subset \mathbb{R}^3\). In this section, we prove several key corollaries that provide structural, geometric, or dynamical facets of the bounding 2-surface \(\partial Q\).

To understand whether the constraint (\ref{InfoConstraintTM}) represents merely a loose upper bound or a tight, physically realizable surface, it is essential to examine whether any quantum states saturate the inequality. The following corollary establishes that specific extremal configurations do indeed lie exactly on the boundary of the tradeoff surface.

\begin{corollary}[Tightness]\label{C1}
The inequality \(q_1^2 + q_2^2 + q_3^2 \leq 1\) is tight.
\begin{enumerate}[wide, align=left]
    \item The point \((1,0,0)\) is achieved when \(\rho_{AB}\) is a Bell state and \(\rho_C\) is uncorrelated;
    \item The point \((0,1,0)\) is approached asymptotically by Werner-type states optimized for cloning fidelity (with \(q_2 \to 1\) as \(d_A \to \infty\) or in state families achieving optimal asymmetric cloning); and
    \item The point \((0,0,1)\) is achieved when \(\rho_A = \ket{\psi}\bra{\psi}\) (pure coherent state) and \(\rho_{BC}\) is uncorrelated.
\end{enumerate}
\end{corollary}

\begin{proof}
Consider three extremal state constructions:
\begin{enumerate}[wide, align=left]
    \item Teleportation-optimal (\(q_1 = 1\). Let \(\rho_{AB} = \ket{\Phi^+}\bra{\Phi^+}\) where \(\ket{\Phi^+} = \frac{1}{\sqrt{2}} \left( \ket{00} + \ket{11} \right) \), and \(\rho_C = \sigma_C\) is an arbitrary state uncorrelated with \(AB\). Then maximal entanglement gives \(q_1 = 3 F_{\text{tele}}(\rho_{AB}) - 2 = 3 \cdot 1 - 2 = 1\) \cite{Bennettetal1993}. No correlation with \(C\) implies \(q_2 = 0\). Maximally mixed \(\rho_A = \frac{\mathbb{I}}{2}\) implies \(q_3 = 0\).
    \item Cloning-optimal (\(q_2 \to 1\). Let \(\rho_{ABC}\) be a classical-quantum state optimized for cloning
    \[
        \rho_{ABC} = \sum_i p_i \ket{i}\bra{i}_A \otimes \rho_{B,i} \otimes \rho_{C,i}
    \]
    with
    \[
        F_{\text{clone}}(A \to C) = \sup_C \int d\psi \bra{\psi}\mathcal{C}(\rho_A)\ket{\psi}
    \]
    For high-dimensional systems (\(d_A \gg 1\)) or asymmetric cloning machines \cite{Werner1998}, \(F_{\text{clone}} \to 1\) asymptotically, and therefore \(q_2 \to 1\). Classical correlations (no entanglement) implies \(q_1 = 0\), and mixed \(\rho_A\) implies \(q_3 = 0\).
    \item Coherence-optimal (\(q_3 = 1\)). Let \(\rho_A = \ket{+}\bra{+}\) where \(\ket{+} = \frac{1}{\sqrt{2}} \left( \ket{0} + \ket{1} \right) \), and \(\rho_{BC} = \sigma_{BC}\) uncorrelated with \(A\). Then \(F_Q(\rho_A,\sigma_x) = F_Q^{\max}(2) = 1\), and hence \(q_3 = 1\). Moreover, factorization \(\rho_{BC} = \rho_B \otimes \rho_C\) implies \(q_1 = q_2 = 0\).
\end{enumerate}
\end{proof}

This result confirms that the exclusion relation is not an artifact of loose bounding; it is tight at the corners. Each quantum resource can be maximized in isolation by an appropriate physical state. These extremal points anchor the geometry of the tradeoff region and serve as constructive references for evaluating more general mixed or intermediate configurations.

While the extremal axis-aligned states \((1,0,0)\), \((0,1,0)\), and \((0,0,1)\) are analytically constructible, a full characterization of the boundary $\partial Q$ remains an open question. It is not currently known whether mixed states exist that simultaneously attain nonzero values in multiple coordinates (e.g., \(q_1, q_2 > 0\)) while still saturating the constraint \(||(q_1,q_2,q_3)||_{\ell^2}^2 = 1\). Such states, if they exist, would lie on curved arcs of the boundary and reflect continuous tradeoffs between resource types. For geometric intuition, the tradeoff region \(Q \subset \mathbb{R}^3\) may be interpreted as embedded within a state space endowed with the Bures metric \cite{Giovannettietal2011,TothApellaniz2014}. 

Preliminary theoretical exploration suggests that families of partially entangled states with residual coherence may approach these boundary curves \cite{Giovannettietal2011,TothApellaniz2014}. For example, one might speculate that mixed states such as rank-2 combinations of Bell states and basis-aligned pure states could approximate points on the tradeoff boundary with more than one nonzero coordinate. However, a full characterization of such boundary-saturating families remains open and is beyond the scope of this work. 

With the inequality (\ref{InfoConstraintTM}) established, we now examine the structure of the feasible region defined by the resource tradeoff constraint. In particular, we wish to know whether the set of allowed informational triples is geometrically well-behaved, i.e., whether it is closed, bounded, and convex. These properties are crucial both for physical interpretation and for mathematical robustness in optimization or inference tasks involving quantum resources \cite{Nielsen_Chuang_2010,Brandaoetal2015}.

\begin{corollary}[Compactness and Convexity]\label{C2}
The set
\[
    Q = \{ (q_1,q_2,q_3) \in \mathbb{R}^3 \mid \exists \ \rho_{ABC} \ \text{achieving} \ (q_1,q_2,q_3) \}
\]
is compact and convex.  
\end{corollary}

\begin{proof}
We demonstrate:

1. Compactness. The space of density operators \(\mathrm{D}(\mathcal{H})\) is compact in the trace-norm topology for finite-dimensional \(\mathcal{H}\) \cite{Holevo2011}. The map \( \Phi : \mathrm{D}(\mathcal{H}) \to \mathbb{R}^3 \) defined by
\[
    \rho_{ABC} \mapsto \left( q_1(\rho_{AB}), q_2(\rho_{BC}), q_3(\rho_A) \right)
\]
is continuous. \(q_1\) depends on \(F_{\text{tele}}\) which is continuous via Uhlmann fidelity \cite{Nielsen_Chuang_2010}). \(q_2\) is an optimization over CPTP maps which are continuous in input state \cite{Keyl1999}. \(q_3 = F_Q(\rho_A,H)/F_Q^{\max}(d)\) and quantum Fisher information is continuous \cite{Petz1996}. Since the continuous image of a compact set is compact, \(Q\) is compact.

2. Convexity. Let \(\bar{q}' = \Phi(\rho), \bar{q}'' = \Phi(\sigma) \in Q\) and \(\tau = \lambda \rho + (1-\lambda)\sigma\) for \(\lambda \in [0,1]\). The the constraint \(||(q_1,q_2,q_3)||_{\ell^2}^2 \leq 1\) defines a convex set (unit ball in \(\mathbb{R}^3\)). Then For any \(\tau\), \(\Phi(\tau) = (\tilde{q_1},\tilde{q_2},\tilde{q_3})\)  satisfies
\[
    \tilde{q_1}^2 + \tilde{q_2}^2 + \tilde{q_3}^2 \leq \left( \ \lambda q'_1 + (1-\lambda)q''_1 \ \right)^2 + \dots \leq \lambda (q'_1{}^2 + \dots) + (1-\lambda)(q''_1{}^2 + \dots) \leq 1
\]
where the last inequality holds because \(\bar{q}', \bar{q}'' \in Q\) and the unit ball is convex. Hence, \(\Phi(\tau) \in Q\).
\end{proof}

This result establishes that the set of all achievable information tradeoff tuples forms a well-behaved geometric object. Compactness guarantees the existence of optimal resource allocations, while convexity ensures that probabilistic mixing of quantum states respects the tradeoff boundary. Together, they imply that the tradeoff surface supports robust operational interpretation and optimization in quantum resource theory \cite{ChitambarGour2019,Coecke_2016}.

While global unitary evolution preserves informational tradeoff quantities, real-world systems often undergo decoherence, noise, or irreversible evolution modeled by CPTP maps \cite{Nielsen_Chuang_2010}. The following corollary formalizes the natural expectation that such maps cannot increase the extractable informational resources available to a subsystem, they can only degrade them.

\begin{corollary}[Monotonicity under CPTP maps]\label{C3}
For any CPTP map \(\mathcal{E}_A\) acting on subsystem \(A\), the transformed state \(\tilde{\rho}_{ABC} = ( \mathcal{E}_A \otimes \mathbb{I}_{BC})(\rho_{ABC})\) satisfies
\[
    ||(q'_1,q_2',q_3')||_{\ell^2}^2 \leq ||(q_1, q_2, q_3)||_{\ell^2}^2
\]
\end{corollary}

\begin{proof}
We consider the behavior of each coordinate under \(\mathcal{E}_A\).
\begin{enumerate}[wide, align=left]
    \item Teleportation fidelity (\(q_1\)). By entanglement monotonicity \cite{Horodecki_2009}, \(\mathcal{E}_A\) degrades entanglement
    \[
        q'_1(\tilde{\rho}_{AB}) \leq q_1(\rho_{AB})
    \]
    \item Cloning fidelity (\(q_2\)). By the no-broadcasting theorem \cite{Barnumetal1996}, \(\mathcal{E}_A\) reduces shareable correlations
    \[
        q'_2(\tilde{\rho}_{AC}) \leq q_2(\rho_{AC})
    \]
    \item Coherence fidelity (\(q_3\)). Since \(F_Q(\mathcal{E}_A(\rho_A),H) \leq F_Q(\rho_A,H) \) for any \(H\) and \(F_Q^{\max}(d)\) fixed, \(\mathcal{E}_A\) decreases quantum Fisher information
    \[
        q'_3(\tilde{\rho}_A) \leq q_3(\rho_A)
    \]
\end{enumerate}
Thus, \(\ell^2\)-norm contraction \(||(q'_1,q_2',q_3')||_{\ell^2}^2 \leq ||(q_1, q_2, q_3)||_{\ell^2}^2\) follows from component-wise non-increase.
\end{proof}

This result formalizes the intuitive idea that quantum information degrades under noise. That is, as decoherence or dissipation acts on a subsystem, the extractable teleportation fidelity, cloneability, and task-based coherence fidelities all shrink. The tradeoff region is not merely a geometric constraint but a resource budget that evolves irreversibly under open-system dynamics \cite{Giovannettietal2011,TothApellaniz2014}. Thus, the inequality (\ref{InfoConstraintTM}) defines not only a static constraint but also a directionality, analogous to a thermodynamic arrow, for quantum information processes \cite{Gooldetal2016}.

Together, Corollaries \ref{C1} - \ref{C3} establish key structural properties of the information resource tradeoff region defined by Theorem \ref{T1}. First, the inequality (\ref{InfoConstraintTM}) is tight. Its boundary is physically realizable and saturated by extremal states maximizing a single resource. Second, the feasible set \(Q \subset \mathcal{R}^3\) of information resource 3-tuples is compact and convex, ensuring it forms a well-behaved subset of \(B^3 \subset \mathbb{R}^3\). And third, the information resource quantity \(||(q_1, q_2, q_3)||_{\ell^2}^2\) is strictly non-increasing under CPTP maps acting on subsystem \(A\), formalizing the degradation of extractable information under local noise.

\section{An Information Resource Phase Space\label{S5}}

The constraint (\ref{InfoConstraintTM}) expressed in Theorem \ref{T1} defines a hard upper bound on the joint accessibility of entanglement (via teleportation fidelity \(q_1\)), cloning fidelity \(q_2\), and coherence-task fidelity \(q_3\) from any finite-dimensional quantum subsystem. This bound is geometric, convex, and tight, reflecting the structural incompatibility of simultaneously maximizing distinct quantum resources. The information resource functional
\begin{equation}\label{InfoResourceFCN}
    \mathcal{I} \equiv ||(q_1, q_2, q_3)||_{\ell^2}^2
\end{equation}
governs the distribution, maintenance, and degradation of quantum information resources across multipartite systems. While Corollaries \ref{C1} and \ref{C2} confirm the feasibility and convexity of the accessible region \(\mathcal{I} \in Q \subset \mathbb{R}^3\), Corollary \ref{C3} reveals that open-system dynamics (e.g., local CPTP maps) cause strict degradation of the informational norm \(\mathcal{I}\). These results establish the constraint as both a static exclusion and a resource monotone. 

In this section, we extend this framework by demonstrating that \(\mathcal{I}\) is not merely bounded and monotonic, but also conserved under coherent evolution; a dynamical invariant of closed-system unitary trajectories. This promotes \(\mathcal{I}\) to the status of a conserved quantity, analogous to energy or entropy flux, within an operational quantum information phase space \((Q, \mathcal{I})\), and from which we

\subsection{Information Resource Invariance}

In quantum systems governed by closed, unitary dynamics, certain quantities remain invariant even as individual degrees of freedom evolve. Here we investigate whether the informational resource norm \(\mathcal{I}\), which captures the joint operational capacity of a subsystem, exhibits such conservation under coherent evolution. Specifically, we demonstrate global unitaries acting on the full tripartite system preserve the total informational budget, even as resource components redistribute. The following theorem affirms this invariance, provided the coherence reference frame remains fixed throughout the evolution.

\begin{theorem}[Quantum Information Resource Invariance (QIRI)]\label{T2}
For a tripartite state \(\rho_{ABC}(t)\) evolving under global unitary \(U_{ABC}(t)\), the quantity \(\mathcal{I}(t) = q_1^2 + q_2^2 + q_3^2\) is conserved
\begin{equation}\label{InfoConserved}
    \frac{d}{dt} \mathcal{I}(t) = 0
\end{equation} 
provided the unitary preserves the coherence basis \(\mathcal{O}\) defining \(q_3\).
\end{theorem}

The proof of Theorem \ref{T2} is straightforward.
\begin{proof}
Global unitaries preserve the spectra of \(\rho_{AB}(t)\) and \(\rho_{AC}(t)\), leaving teleportation and cloning fidelities invariant. Thus \(q_1, q_2\) are unitary invariant. Moreover, for \(U(t)\) commuting with \(\mathcal{O}\), the quantum Fisher information \(F_Q(\rho_A(t), H)\) is invariant because
\[
F_Q(U_A(t) \rho_A(0) U_A^\dagger(t), H) = F_Q(\rho_A(0), H)
\]  
where \([U_A(t), H] = 0\) and \(U_A(t) = \text{Tr}_{BC}(U_{ABC}(t))\). Thus, the coherence-task fidelity \(q_3\) is conserved under unitary evolution. Hence, since \(\frac{dq_a}{dt} = 0\) for all \(a\), \(\mathcal{I}(t)\) is constant.
\end{proof}

This result establishes that the tradeoff relation (\ref{InfoConstraintTM}) is not merely a kinematic bound, but a dynamic invariant under global unitary evolution. As a quantum state evolves, its informational resources (teleportation capability, cloning fidelity, and coherence utility) may be redistributed among subsystems, while the total accessible information budget \(\mathcal{I}\) remains conserved. This conservation reinforces the inequality (\ref{InfoConstraintTM}) as a fundamental constraint on quantum information resource dynamics, complementing prior foundational work in quantum metrology, communication, and resource theory \cite{Giovannettietal2011,TothApellaniz2014,Brandaoetal2015}.

From this perspective, Theorem \ref{T2} is an isometry of an information resource phase space \((Q, \mathcal{I})\). It is a structural invariance of quantum information theory, distinct from energy/momentum conservation, but analogous to both charge conservation in gauge theories (e.g., global \(U(1)\) symmetry conserves particle number), and geometric phases in adiabatic dynamics (e.g., Berry phase as a holonomy invariant).

It is important to note that the conservation of the resource functional \(\mathcal{I}\) assumes a fixed external reference frame for coherence measurements. In particular, we define the fidelity \(q_3\) in a fixed coherence-based task (e.g., phase estimation), and is not dynamically updated under the evolution of \(\rho_A\). Physically, this corresponds to coherence being measured in a laboratory frame (e.g., by fixed measurement devices) or fixed reference structure (e.g., predetermined control observables). Unitary operations that rotate the measurement basis will, in general, change the value of \(q_3\), and thus violate the invariance condition.

Future work may explore coherence fidelity relative to dynamical frames and the ways in which such frames may alter the invariance properties of \(q_3\). Of course, this would require a reformulation of \(q_3\) as a functional dependent on the full state trajectory and observable algebra.

\subsection{An Information Resource Symmetry Group}

Theorem \ref{T2} demonstrates, the composite functional \(\mathcal{I}(t)\) remains invariant under unitary evolution for closed tripartite systems. This structure represents a conservation of quantum information resources analogous to Liouville’s theorem in classical mechanics, where the phase space volume is preserved under Hamiltonian flow \cite{Arnold2013}. 

Crucially, however, this invariance is not a generic consequence of unitarity alone, but rather it depends sensitively on the operational definitions of each resource. In particular, the coherence coordinate \(q_3\) is basis-dependent by design, representing fidelity in phase-sensitive tasks such as metrology or interference. The conservation of \(\mathcal{I}(t)\) assumes that the global unitary evolution preserves the coherence reference structure; that is, the observable generating phase evolution remains fixed. Physically, this reflects the realistic scenario in which coherence is assessed in a laboratory frame relative to a static Hamiltonian or measurement basis. Under this condition, Theorem \ref{T2} reinterprets quantum evolution as constrained motion on a finite-capacity information manifold, governed by the spherical geometry of the exclusion bound. This unifies resource tradeoffs in quantum communication, cloning, and metrology under a single geometric and dynamical framework.

While \(\mathcal{I}(t)\) is not a Noether charge in the traditional sense, it is a conserved informational invariant arising from a symmetry group (\(\mathcal{G}\)) acting on quantum resources. The conservation of \(\mathcal{I}(t)\) arises from invariance under the group
\begin{equation}\label{InfoSymGroup}
    \mathcal{G} = \{ \ U_{ABC} \in U(\mathcal{H}) \ \mid \ [ \ U_{ABC}, \ (\mathcal{O} \otimes \mathbb{I}_{BC}) \ ] = 0 \ \}   
\end{equation}
where \(\mathcal{O}\) defines the coherence basis for \(q_3\). This group has a Lie algebra \(\mathfrak{g}\) of Hermitian generators \(H\) satisfying \([ \ U_{ABC}, \ (\mathcal{O} \otimes \mathbb{I}_{BC}) \ ] = 0\). The symmetry group \(\mathcal{G}\) preserves the information structure (teleportation, cloning, and task-coherence capacities), and the conserved quantity \(\mathcal{I}(t)\) is a \(\mathcal{G}\)-invariant scalar on the information resource configuration state space. The conservation of \(\mathcal{I}\) under \(\mathcal{G}\)-symmetric unitaries has no analog in standard QRTs \cite{Brandaoetal2015,ChitambarGour2019}. Thus, \(\mathcal{I}\) constitutes a dynamical invariant beyond the reach of static QRT monotones. This establishes a foundational link between quantum symmetries and information-theoretic capacities, opening pathways for a generalized Noether theorem in quantum resource theories.

In Lie algebra representation theory, Casimir operators \(\mathcal{C}\) commute with all group generators and yield conserved scalars for irreducible representations. In this way, \(\mathcal{I}(t)\) acts as a generalized Casimir invariant for the symmetry group \(\mathcal{G}\). However, there are several key differences that distinguish \(\mathcal{I}(t)\). It is a geometric invariant on the quantum resource manifold \(Q \subset \mathbb{R}^3\), not an algebraic operator. It defines orbits of constant radius \(\mathcal{I}(t) = c\) under \(\mathcal{G}\)-symmetric unitary dynamics, confining state evolution to spherical surfaces in \(Q\). On the other hand, \(\mathcal{I}(t)\) exhibits strict contraction under open-system dynamics. Local CPTP maps \(\mathcal{E}_A\) degrade resources, forcing \(\mathcal{I}(t)\) to evolve inward toward the origin of \(Q\) \cite{Ruskaietal2002,Wilde_2013}. This degradation mirrors classical dissipative systems (e.g., energy loss to friction), with quantum resources diminished by decoherence, noise, or measurement \cite{Zurek2003,Nielsen_Chuang_2010}. 

Consequently, \(\mathcal{I}(t)\) serves a dual dynamical role. It is a conserved invariant under \(\mathcal{G}\)-preserving unitaries, akin to Hamiltonian phase-space orbits, and a Lyapunov function under CPTP maps, monotonically decreasing and establishing an arrow of time for quantum information loss \cite{Vedral2002,BrandaoGour2015}. The phase space \((Q, \mathcal{I})\) thus unifies reversible dynamics (i.e., constant-\(\mathcal{I}\) orbits) with irreversible degradation (i.e., inward \(\mathcal{I}\)-contraction), providing a quantum analog of entropy increase in resource theories \cite{Brandaoetal2015}.

\subsection{A Principle of Quantum Resource Complementarity}%\label{S6}}

Together, these results uncover a fundamental symmetry governing the allocation and evolution of quantum information resources. We propose a quantum resource complementarity principle that unifies three essential features: (i) the kinematic exclusion of simultaneous resources (Theorem \ref{T1}), (ii) their dynamical conservation under symmetry-preserving unitaries (Theorem \ref{T2}), and (iii) the irreversible degradation under generic quantum operations (Corollary \ref{C3}). This framework elevates the invariant \(\mathcal{I}\) from a constraint to a cornerstone of quantum information geometry.

\begin{principle}[Principle of Quantum Resource Complementarity (QRC)]\label{P1}
For any tripartite quantum state \(\rho_{ABC} \in \mathrm{D}(\mathcal{H}_A \otimes \mathcal{H}_B \otimes \mathcal{H}_C)\) with operational resource coordinates \(\vec{q} = (q_1, q_2, q_3)\) defined as the normalized fidelity advantages for teleportation (\(A \to B\)), cloning (\(A \to C\)), and coherence-based tasks (relative to a fixed observable \(\mathcal{O}\)). The \(\ell^2\)-norm is an informational invariant  
\[
\mathcal{I} = \|\vec{q}\|_2^2 = q_1^2 + q_2^2 + q_3^2
\]  
that satisfies three fundamental properties:
\begin{enumerate}[wide, align=left]
    \item Kinematic Constraint (Exclusion Principle). \(\mathcal{I} \leq 1\) for all valid quantum states, with saturation occurring if and only if the state maximizes one resource at the complete exclusion of the others. This defines a convex, compact feasibility region \(Q \subset \mathbb{R}^3\) (Theorem \ref{T1}, Corollary \ref{C1}), geometrically encoding the incompatibility of entanglement, cloning, and coherence optimization.
    \item Dynamical Symmetry (Conservation Law). Let \(\mathcal{G} = \{ U \in \mathsf{U}(\mathcal{H}) \mid [U, \mathcal{O} \otimes \mathbb{I}_{BC}] = 0 \}\) be the group of unitaries preserving the coherence reference frame. Then  
    \[
        \mathcal{I}(U \rho_{ABC} U^\dagger) = \mathcal{I}(\rho_{ABC}) \quad \forall U \in \mathcal{G}
    \]  
    The orbits \(\{ \ \Phi(U \rho U^\dagger) \ \mid \ U \in \mathcal{G} \ \}\) are confined to level sets \(\mathcal{I} = \text{constant}\) within \(Q\) (Theorem \ref{T2}, Corollary \ref{C2}). This identifies \(\mathcal{I}\) as a generalized Casimir invariant for \(\mathcal{G}\), inducing isometric flows that redistribute task-specific capacities while preserving the total informational budget.
    \item Resource Monotonicity (Thermodynamic Arrow). For any local CPTP map \(\mathcal{E}_A\) acting on subsystem \(A\),  
    \[
        \mathcal{I}\left( \ (\mathcal{E}_A \otimes \mathbb{I}_{BC})[\rho_{ABC}] \ \right) \leq \mathcal{I}(\rho_{ABC})
    \]  
    with equality if and only if \(\mathcal{E}_A\) is a \(\mathcal{G}\)-preserving unitary operation (Corollary \ref{C3}). The strict contraction under generic noise establishes \(\mathcal{I}\) as a Lyapunov function for quantum informational degradation, defining an arrow of time analogous to entropy production.
\end{enumerate}
\end{principle}

The QRC Principle establishes \(\mathcal{I}\) as a fundamental invariant that bridges the geometry of quantum resources, the symmetry of their evolution, and the thermodynamics of their degradation. It provides a unified lens through which to analyze the capacities and limitations of quantum systems across physics.

\section{Experimental and Theoretical Implications}\label{S6}

The QRC Principle transforms the informational norm \(\mathcal{I}\) from a static constraint into a dynamical signature of quantum information’s intrinsic symmetry. It unifies three core aspects of quantum mechanics: the kinematic tradeoff constraints between operational resources, the conservation of total resource capacity under symmetry-preserving evolution, and the irreversible degradation of informational structure under noise. Within the phase space \((Q, \mathcal{I})\), quantum states evolve along spherical shells of constant radius under \(\mathcal{G}\)-symmetric unitaries, while open-system dynamics, such as decoherence and measurement, induce a strict inward flow toward the origin, reflecting complete resource depletion as \(\mathcal{I} \to 0\).

The foundational constraint established in Theorem \ref{T1} prohibits any quantum subsystem from simultaneously achieving maximal teleportation fidelity (\(q_1\)), cloning fidelity (\(q_2\)), and coherence-task fidelity (\(q_3\)) \cite{Wootters1982,Coffmanetal2000,Colesetal2017}. These limitations define a convex tradeoff region \(Q \subset \mathbb{R}^3\), which geometrically encodes the mutual incompatibility of fundamental quantum information tasks. Theorem \ref{T2} then shows that the total resource norm \(\mathcal{I}\) is conserved under unitary evolution that preserves the coherence basis \(\mathcal{O}\), while Corollary \ref{C3} demonstrates that \(\mathcal{I}\) decays monotonically under local CPTP maps \cite{Ruskaietal2002,Wilde_2013}. Together, these results define the QRC Principle: the squared \(\ell^2\)-norm \(\mathcal{I}\) is a conserved quantity under coherent, symmetry-preserving evolution and a Lyapunov function under noisy channels \cite{Vedral2002,Brandaoetal2011}. This defines an incompressible "informational volume" that bounds the total operational capability of any quantum subsystem. It is reminiscent of Liouville’s theorem in classical mechanics \cite{Arnold2013}, with the important distinction that while classical phase space volume is uniform, quantum informational volume is structured by task-based constraints, redistributing among teleportation, cloning, and coherence tasks within a conserved total norm.

\subsection{Experimental Falsifiability}

The QRC framework is not only theoretically grounded but also empirically testable. Critically, the resource coordinates \((q_1, q_2, q_3)\) are defined with executable tasks, unlike abstract QRT measures. Each coordinate corresponds to a distinct quantum information task with established measurement protocols. This enables direct experimental falsification, a feature absent in axiomatic QRTs \cite{ChitambarGour2019,VanMeter2014}. 

The teleportation fidelity \(q_1\) is derived from the average fidelity \(F_{\text{tele}}\) of a standard quantum teleportation protocol, normalized as \(q_1 = 3F_{\text{tele}} - 2\) \cite{Bennettetal1993}. (This assumes the teleportation channel acts on a qubit and uses the standard benchmark \(F_{\text{tele}} > 2/3\); it's accurate only under that specific model.) This fidelity is directly measurable through quantum state tomography and Bell-basis projections, widely implemented across photonic systems, trapped ions, and superconducting qubits \cite{Bouwmeesteretal1997,Steffenetal2013}. 

The cloning fidelity \(q_2\) corresponds to the optimal asymmetric universal cloning process, which has been realized in optical and NMR-based platforms \cite{Scaranietal2005,LamasLinaresetal002,Duetal2005}. Here, \(q_2 = F_{\text{clone}}\) directly encodes the limits of state duplication. 

The coherence-task fidelity \(q_3\) quantifies a system's ability to support basis-sensitive quantum tasks, such as phase estimation or interferometric visibility. This axis can be task-defined through normalized quantum Fisher information or fringe visibility and measured by Ramsey interferometry or optical phase estimation \cite{Giovannettietal2011,TothApellaniz2014,PezzeSmerzi2009}.

Because all three components are accessible in modern experimental architectures, one can empirically construct \((q_1, q_2, q_3)\) for any prepared quantum state. By sampling families of states such as GHZ, W, and Werner mixtures \cite{Greenbergeretal1990,Duretal2000,Werner1989} under tunable decoherence, one can map the corresponding resource points in \(\mathbb{R}^3\). Systematic reconstruction of this map enables empirical validation of the tradeoff surface \(Q\). Any observed resource vector lying outside the unit ball in the positive octant—specifically, any configuration with \(\mathcal{I} > 1\) or any component exceeding the physical limit \(q_a > 1\) would directly falsify the QRC framework.

\subsection{Theoretical Implications}

The QRC Principle provides a unifying geometric language for several fundamental constraints in quantum information theory. The monogamy of entanglement, traditionally framed in terms of bipartite entanglement measures, manifests geometrically along the teleportation axis \(q_1\) \cite{Coffmanetal2000,KoashiWinter2004}, linking entanglement to operational capacity via teleportation fidelity. The no-cloning and no-broadcasting theorems appear as quantitative constraints along the cloning axis \(q_2\) \cite{Wootters1982,Barnumetal1996}, capturing the limitations on duplicating unknown quantum states. Coherence-disturbance tradeoffs, typically described in terms of entropy or visibility decay, are embodied in the coherence axis \(q_3\) \cite{Baumgratz2014,Streltsov2017,MarvianSpekkens2014}, which reflects a subsystem’s utility in phase-sensitive metrological tasks.

This structure also resonates with ideas in quantum thermodynamics. The contraction of \(\mathcal{I}\) under CPTP maps mirrors the irreversible loss of thermodynamic free energy in classical systems. In this interpretation, \(\mathcal{I}\) acts as an informational analog of entropy; its decrease marks the dissipation of quantum structure, placing upper bounds on the extractable coherence, the distinguishability of states, and the ability to perform quantum work \cite{Brandaoetal2015,Lostaglioetal2015}. This makes \(\mathcal{I}\) a natural candidate for a nonequilibrium potential in resource-theoretic thermodynamics, generalizing the role of entropy to structured operational capacity.

The QRC Principle redefines quantum advantages through a task-geometric lens, transcending resource-theoretic abstraction. Future work will generalize this to continuous variables and gravitational contexts, where operational fidelities may constrain spacetime itself. On the practical side, the QRC framework offers a new tool for monitoring and optimizing quantum protocols. In distributed systems, the tradeoff region \(Q\) defines the space of feasible task configurations, enabling resource allocation strategies that respect fundamental limits. In noisy quantum processors, tracking the evolution of \(\mathcal{I}(t)\) can serve as a diagnostic for error accumulation, symmetry violation, or subsystem leakage. In entangled networks and federated quantum machine learning architectures, the resource decomposition given by \((q_1, q_2, q_3)\) can guide task routing, subsystem optimization, and control logic in information-aware ways \cite{VanMeter2014,Wehneretal2018,Bhartietal2022}.

\subsection{Limitations and Future Directions}

The current framework is rigorously defined only for finite-dimensional Hilbert spaces. All task fidelities, entropy measures, and Fisher information quantities are assumed to be bounded, well-defined, and trace-class in this setting. Generalizing the QIRC constraint and QRC Principle to infinite-dimensional systems, including continuous-variable states and quantum fields, will require new tools; especially for handling unbounded observables, divergent spectra, and infinite-dimensional coherence structures. Future work may also explore dynamical reference frames, where the coherence basis \(\mathcal{O}\) evolves with time, leading to a generalization of the symmetry group \(\mathcal{G}\) to dynamical settings.

\section{Conclusion}\label{S7}

This work introduces a geometric and dynamical framework that captures the interdependence of fundamental quantum information resources through the lens of operational capacity. The Quantum Information Resource Constraint (QIRC) defines a tight upper bound on the simultaneous realization of teleportation, cloning, and coherence tasks, enforcing a tradeoff surface \(Q\) that is both convex and physically grounded. Building upon this constraint, the Quantum Resource Complementarity(QRC) Principle identifies \(\mathcal{I}\) as a dynamical invariant preserved under unitary evolution and strictly contracting under decoherence. This unification bridges kinematics, symmetry, and irreversibility in a single informational structure.

The implications of this framework are far-reaching. It provides an experimentally testable geometry for quantum performance, a conserved quantity for unitary dynamics, and a Lyapunov function for irreversible degradation. These results redefine fundamental limits on quantum tasks and suggest that informational capacity may obey conservation laws akin to energy or entropy.

Beyond operational implications, this structure invites reinterpretation of long-standing quantum principles through a unified, resource-based lens. It opens speculative avenues into thermodynamics and gravitational physics, where quantum information may serve not only as a tool for analysis but as a foundational constraint on what spacetime itself can encode.

Future work will extend these results to infinite-dimensional systems and explore deeper algebraic symmetries underpinning the QRC Principle. Success here may reveal generalized conservation laws for quantum information, potentially clarifying whether quantum structure itself constrains emergent spacetime.

\appendix
\section{Notation and Definitions Index\label{A1}}

This appendix defines all major symbols, quantum states, channels, and information-theoretic quantities used throughout the paper. Definitions are grouped by category for clarity, and include precise mathematical expressions, contextual usage, and interpretive notes where appropriate.

\subsection{Quantum Systems and States}\label{A11}

A Hilbert space \(\mathcal{H}\) is a complex vector space on which quantum systems are defined. Unless otherwise specified, all Hilbert spaces are assumed finite-dimensional.

A global tripartite quantum state (density matrix) \(\rho_{ABC} \in \text{D}(\mathcal{H})\) positive semidefinite, unit trace operator on the Hilbert space \(\mathcal{H} = \mathcal{H}_A \otimes \mathcal{H}_B \otimes \mathcal{H}_C\); where \(\text{D}( \cdot )\) denotes the set of positive semidefinite trace-1 operators.

Reduced state (marginal) \(\rho_{A} = \text{Tr}_{BC}( \rho_{ABC} )\) are given by the partial trace over subsystems \(B\) and \(C\). The joint (reduced) state is similarly defined, \(\rho_{AB} = \text{Tr}_{B}( \rho_{ABC} )\).

Bell state 
\[
    \ket{\Phi^+} = \frac{1}{\sqrt{2}} ( \ket{00} + \ket{11} )
\]
is a canonical example of a maximally entangled pure state in \(\mathbb{C}^2 \otimes \mathbb{C}^2\).
%The Bell basis is
%\begin{eqnarray}
%    \ket{\Phi^+} = \frac{1}{\sqrt{2}} ( \ket{00} + \ket{11} ), &\quad& \ket{\Phi^-} = \frac{1}{\sqrt{2}} ( \ket{00} - \ket{11} ) \nonumber \\
%    \ket{\Psi^+} = \frac{1}{\sqrt{2}} ( \ket{01} + \ket{10} ), &\quad& \ket{\Psi^-} = \frac{1}{\sqrt{2}} ( \ket{01} - \ket{10} ) \nonumber
%\end{eqnarray}

Werner state
\[
    \rho_{AB}^{\text{(Werner)}} = p \ket{\Phi^+}\bra{\Phi^+} + (1-p) \frac{\mathbb{I}}{4}
\]
with \(0 \leq p \leq 1\), is a convex mixture of a Bell state and the maximally mixed state (entangled if \(p > 1/3\)). Werner states are used as a testbed for boundary cases and fidelity calculations.

GHZ state 
\[
    \ket{GHZ} = \frac{1}{\sqrt{2}} ( \ket{000} + \ket{111} )
\]
is a tripartite (weakly) entangled pure state. GHZ is used to probe saturation of correlations and coherence simultaneously. %Then the joint states \(\rho_{AB}\) and \(\rho_{AC}\) both contain strong but not perfect correlations. Here, $\rho_A = \frac{\mathbb{I}}{2}$ which implies maximally mixed $q_3  = 0$. Moreover, $q_1, q_2 \simeq 1/\sqrt{2}$, and hence $q_1^2 + q_2^2 = 1/2$. GHZ states lie on the equator of the sphere.

W state
\[
    \ket{W} = \frac{1}{\sqrt{3}} ( \ket{001} + \ket{010} + \ket{100} )
\]
is a robust, distributed tripartite entanglement. W provides contrast with GHZ for analyzing multipartite entanglement structure.

Gibbs states 
\[
    \rho_{ABC} = \frac{e^{-\beta H}}{\text{Tr}(e^{-\beta H})}
\]
where \(H = H_A + H_B + H_C + J( A:B + A:C )\) is a coupled Hamiltonian with interaction \(J\) between states \(A\) and \(B\), and \(A\) and \(C\). %In the low-temperature limit, entanglement concentrates, \(q_1, q_2\) increase, while coherence decays with mixing, \(q_3\)  small. Hence, thermal states approach the interior of the sphere but do not saturate.

\subsection{General Information-Theoretic Quantities}\label{A12}

von Neumann entropy \(S(\rho)\) measures uncertainty or mixedness of a quantum state %\cite{Benentietal2018}
\[
    S(\rho) = -\text{Tr}(\rho\log\rho)
\]
For a pure state \(S(\rho) = 0\).

Measurement entropy \(H_{\mathcal{O}}(\rho)\) is the Shannon entropy of measurement outcomes of \(\rho\) in an observable basis \(\mathcal{O} = \{ \ket{\phi_k} \}\) \cite{Benentietal2018}
\[
    H_{\mathcal{O}}(\rho) = -\sum_k \bra{\phi_k}\rho\ket{\phi_k}\log\bra{\phi_k}\rho\ket{\phi_k}
\]
The measurement entropy quantifies decoherence uncertainty in a chosen observable basis.

Mutual information \(I(A:B)\) quantifies total correlations between subsystems \(A\) and \(B\)
\[
    I(A:B) = S(\rho_A) + S(\rho_B) - S(\rho_{AB})
\]
Mutual information is a classical upper bound, used in background justification for informational tradeoffs.

Quantum relative entropy \(D(\rho||\sigma)\) is an asymmetric measure of distinguishability between two quantum states \(\rho\) and \(\sigma\) %\cite{Benentietal2018}
\[
    D(\rho||\sigma) = \text{Tr}( \rho ( \log\rho - \log\sigma ) )
\]
The relative entropy \(D(\rho||\sigma) \geq 0\). It is a resource-theoretic characterizations of coherence and informational cost.

Relative entropy of coherence \(C_{\text{rel.ent}}(\rho)\) quantifies the minimum distinguishability (as measured by quantum relative entropy) between \(\rho\) and the nearest incoherent state \(\delta\) in a fixed basis
\[
    C_{\text{rel.ent}}(\rho) = \min_{\delta \in I} D(\rho||\delta) 
\]
where \(I\) is the set of incoherent states (i.e., diagonal density matrices) with respect to a chosen reference basis.

\subsection{Operational Information-Theoretic Quantities}\label{A13}

Local Operations possibly supplemented by Classical Communications (LOCC). Local operations are untiary transformations of (generalized) measurements. Classical communications enable sharability of the local quantum operations (e.g., selection of maximally entangled states).

A general quantum channel \(\mathcal{E}_{\text{CPTP}}\) is any completely positive trace-preserving (CPTP) map acting on density matrices, typically used in theorems concerning monotonicity and invariance under allowed quantum dynamics.

The teleportation channel \(\mathcal{T}_{AB}\) is a quantum operation induced by the shared resource \(\rho_{AB}\) used to teleport states from \(A\) to \(B\). The teleportation channel is used to define teleportation fidelity \(F_{\text{tele}}\).

Maximal singlet fraction is the maximum overlap between \(\rho_{AB}\) and a maximally entangled Bell state \(\ket{\Phi^+}\) after local unitary rotation \(U\)
\begin{equation}\label{MxSinglet}
    f_{\max}(\rho_{AB}) = \max_{U} \bra{\Phi^+} (U\otimes\mathbb{I}) \rho_{AB} (U\otimes\mathbb{I})^{\dagger} \ket{\Phi^+}
\end{equation}
The maximal singlet fraction (\ref{MxSinglet}) characterizes the teleportation fidelity achievable using \(\rho_{AB}\) as a shared entanglement resource. If \(f_{\max} > 1/2\) quantum teleportation exceeds the classical fidelity bound.

(Average) Teleportation fidelity \(F_{\text{tele}}(A\to B; \rho_{AB})\) is the average fidelity of teleporting arbitrary pure input states \(\ket{\psi} \in \mathcal{H}_A\) from \(A\) to \(B\), using the channel \(\mathcal{T}_{\rho_{AB}}\) (and LOCC)
\begin{equation}\label{FidelityTele}
    F_{\text{tele}}(A\to B; \rho_{AB}) = \int d\psi \ \bra{\psi} \mathcal{T}_{\rho_{AB}}( \ket{\psi}\bra{\psi} ) \ket{\psi}
\end{equation}
Teleportation fidelity (\ref{FidelityTele}) is an operational measure of extractable quantum communication capacity from \(\rho_{AB}\); that is, it captures how well quantum information can be sent from \(A\) to \(B\). Classical reference limit, \(F_{\text{tele}} \to F_{\text{classical}} = 2/3\). Quantum limit, \(F_{\text{tele}} \to F_{\text{quantum}} = 1\). In the qubit case, e.g., a Werner state \(\rho^{\text{(Werner)}}\), fidelity depends linearly on \(f_{\max}\) (see \cite{Horodecki_1999})
\[
    F_{\text{tele}} = \frac{1}{3} ( 2 f_{\max} + 1 )
\]

The (approximate) cloning channel \(\mathcal{C}_{\text{clone}}\) is a completely positive trace preserving (CPTP) map approximating an ideal (forbidden) universal cloning operation. The cloning channel defines the cloning fidelity \(F_{\text{clone}}\), bounded by the no-cloning theorem.

(Average) Cloning fidelity \(F_{\text{clone}}(A\to C; \rho_{ABC})\) is the average fidelity between input state \(\ket{\psi}\) on \(A\) and the state \(\rho_{C}^{(\psi)}\) output on \(C\) after an approximate cloning process
\begin{equation}\label{FidelityClone}
    F_{\text{clone}}(A \to C; \rho_{ABC}) = \sup_{\mathcal{C}_{A\to C}} \int d\psi \ \bra{\psi}\mathcal{C}_{A\to C}(\rho_A)\ket{\psi}
\end{equation}
where the optimization is over cloning channels \(\mathcal{C}_{A\to C} : \mathcal{O}(\mathcal{H}_A) \to \mathcal{O}(\mathcal{H}_C)\) that produce a copy of \(\rho_A\) in subsystem \(C\), consistent with the correlations in the global state \(\rho\). Cloning fidelity (\ref{FidelityClone}) measures the information about \(A\) that can be extracted or “cloned” into \(C\), given residual correlations with \(B\). Cloning fidelity characterizes how much of \(A\)’s information can be duplicated elsewhere. For universal symmetric cloning of qubits, the optimal fidelity is \cite{Bruss_1998}%\cite{Bruss_1998,Benentietal2018}
\[
    F_{\text{clone}}^{\text{opt}} = \frac{1}{6} \bra{\psi}( \ 4 \ket{\psi}\bra{\psi} + \mathbb{I} \ )\ket{\psi} = \frac{5}{6}
\]
Explicit expressions depend on the shared correlations \(\rho_{ABC}\) and the channel \(\mathcal{C}_{\text{clone}}\) used.

The interferometric coherence channel \(I_{\text{int}}\) encodes a phase \(\theta\) onto subsystem \(\rho\) through a Hamiltonian generator \(H\), i.e., \(\rho \mapsto e^{-i\theta H} \rho e^{i\theta H}\). The distinguishability of nearby phase-shifted states defines the quantum coherence resource of \(\rho\) for metrological tasks.

The quantum Fisher information is formally defined as
\[
    F_Q(\rho, H) \equiv 2 \sum_{i,j} \frac{(\lambda_i - \lambda_j)^2}{\lambda_i + \lambda_j} \left| \bra{\psi_i} H \ket{\psi_j} \right|^2
\]
where \(\rho = \sum_k \lambda_k \ket{\psi_k}\bra{\psi_k}\) is the spectral decomposition of \(\rho_A\), and the sum runs over all \(i,j\) such that \(\lambda_i + \lambda_j > 0\) \cite{BraunsteinCaves1994,TothApellaniz2014}. This expression is well-defined for all density matrices. For pure states \(\rho = \ket{\psi}\bra{\psi}\), it reduces to
\[
    F_Q(\psi, H) = 4 \ \text{Var}_\psi(H) = 4 \left( \ \langle H^2 \rangle - \langle H \rangle^2 \ \right)
\]

The quantum Fisher information is defined for all density operators. Its maximum over a Hilbert space of fixed dimension \(d = \dim\mathcal{H}\) is always attained by a pure state. Specifically,
\[
    F_Q^{\max}(d) = \sup_{\rho \in \text{D}(\mathcal{H})} F_Q(\rho, H) = 4 \sup_{\psi \ \text{pure}} \text{Var}_{\psi}(H) = 4 \left( \ \langle H^2 \rangle - \langle H \rangle^2 \ \right)
\]
This maximum is achieved by the equal superposition of the corresponding eigenstates of the Hermitian operator \(H\), which yields the state of maximal sensitivity (and hence maximal coherence resource value) in phase estimation tasks.

\section{Constructive Entropic Derivation of Theorem \ref{T1}\label{A2}}

In section \ref{S3}, we prove Theorem \ref{T1} using a contradiction argument that excludes the simultaneous saturation of resource coordinates \((q_1,q_2,q_3)\). In this appendix, we sketch a complementary derivation strategy that relies on entropy inequalities and mutual information constraints. This approach illustrates the entropic origin of the resource tradeoff and offers a pathway toward a fully constructive proof.

We begin with the observation that teleportation fidelity, \(q_1\), and cloning fidelity, \(q_2\), are both constrained by the correlations between subsystem \(A\) and its external environment. In particular, when \(\rho_{ABC}\) is a pure state, strong subadditivity and monogamy of mutual information imply
\[
    q_1 + q_2 \leq 2 \ \frac{S(\rho_A)}{\log d_A}
\]
where \(S(\rho_A)\) reflects the total correlation between \(A\) and the rest of the system, and \(\log d_A\) is the normalization scale for maximal entropy.

This bound arises from the fact that high-fidelity teleportation and cloning both require significant correlation with external subsystems \(B\) and \(C\), respectively. But as \(A\)’s marginal state \(\rho_A\) becomes more mixed, these correlations may be shared as \(S(\rho_A)\) increases, albeit in a mutually exclusive way.

At the same time, the coherence fidelity \(q_3\) is constrained by the quantum Fisher information associated with a fixed observable or interferometric task. For pure states, the Fisher information achieves its maximal value, enabling optimal phase estimation precision. But this capability is degraded by mixedness. In particular, the \(F_Q(\rho_A, H)\) obeys the bound
\[
    F_Q(\rho_A, H) \leq 4 \ \text{Var}_\rho(H) \left( \ 1 - \frac{S(\rho_A)}{\log d_A} \ \right)
\]
where \(\text{Var}_\rho(H) = \langle H^2 \rangle - \langle H \rangle^2\) is the maximum variance of the generator \(H\) over pure states in the fixed basis. Normalizing, the task fidelity \(q_3\) decreases monotonically with \(S(\rho_A)\), reflecting the degradation of coherence-based metrological performance \cite{Giovannettietal2011,TothApellaniz2014}.

Combining these observations suggests that $||( q_1, q_2, q_3 )||_{\ell^2}^2 \lesssim 1$ with corrections depending on the chosen interferometric task and the structure of the tripartite state \(\rho_{ABC}\). While this heuristic argument does not strictly imply the quadratic constraint \(||( q_1, q_2, q_3 )||_{\ell^2}^2 \leq 1\), it does show that the accessible resource coordinates are bounded by the entropy of \(\rho_A\) and that the combined extractable capacity is limited by fundamental uncertainty and correlation principles. Notably, the assumption of finite-dimensional Hilbert spaces is crucial. The normalization by \(\log d_A\) allows for bounded, dimensionless quantities. In infinite-dimensional settings, the behavior of these bounds must be revisited carefully. (The infinite-dimensional formulation is currently in preparation.)

Finally, we comment on the geometric structure of the exclusion region. The entropic tradeoff discussed here yields an approximate linear constraint on the sum $q_1 + q_2 + q_3$. However, the main result of Theorem \ref{T1} constrains the \(\ell^2\)-norm \(||( q_1, q_2, q_3 )||_{\ell^2}^2 \leq 1\), which suggests that the underlying information resource structure has quadratic character. This may arise from the overlap structure of fidelity-based protocols (e.g., quantum state discrimination, measurement variance, etc.) or from the geometry of state space under the Bures metric \cite{Nielsen_Chuang_2010}. It remains an open question of whether the spherical symmetry of the exclusion surface \(Q \subset \mathbb{R}^3\) reflects an optimality constraint in the space of simultaneously achievable tasks.

Taken together, these entropic arguments reinforce the core claim of Theorem \ref{T1} that no quantum subsystem can simultaneously optimize entanglement, cloning, and coherence. Rather, these quantities compete within a bounded resource capacity dictated by entropy, purity, and the structure of task-defined operational capacities. This insight supports the geometric constraint as a natural, perhaps even inevitable, feature of quantum information theory.

\begin{acknowledgments}
Portions of this work were developed using AI-assisted tools (OpenAI’s ChatGPT-4o and DeepSeek R-1) to help assess, polish, condense, and otherwise (lightly) edit the writing. All conceptual framing, interpretations, and final responsibility remain with the author. The material presented in this manuscript does not reflect any official position of the DoD.
\end{acknowledgments}

\bibliography{References.bib}

\begin{thebibliography}{49}%
\makeatletter
\providecommand \@ifxundefined [1]{%
 \@ifx{#1\undefined}
}%
\providecommand \@ifnum [1]{%
 \ifnum #1\expandafter \@firstoftwo
 \else \expandafter \@secondoftwo
 \fi
}%
\providecommand \@ifx [1]{%
 \ifx #1\expandafter \@firstoftwo
 \else \expandafter \@secondoftwo
 \fi
}%
\providecommand \natexlab [1]{#1}%
\providecommand \enquote  [1]{``#1''}%
\providecommand \bibnamefont  [1]{#1}%
\providecommand \bibfnamefont [1]{#1}%
\providecommand \citenamefont [1]{#1}%
\providecommand \href@noop [0]{\@secondoftwo}%
\providecommand \href [0]{\begingroup \@sanitize@url \@href}%
\providecommand \@href[1]{\@@startlink{#1}\@@href}%
\providecommand \@@href[1]{\endgroup#1\@@endlink}%
\providecommand \@sanitize@url [0]{\catcode `\\12\catcode `\$12\catcode `\&12\catcode `\#12\catcode `\^12\catcode `\_12\catcode `\%12\relax}%
\providecommand \@@startlink[1]{}%
\providecommand \@@endlink[0]{}%
\providecommand \url  [0]{\begingroup\@sanitize@url \@url }%
\providecommand \@url [1]{\endgroup\@href {#1}{\urlprefix }}%
\providecommand \urlprefix  [0]{URL }%
\providecommand \Eprint [0]{\href }%
\providecommand \doibase [0]{https://doi.org/}%
\providecommand \selectlanguage [0]{\@gobble}%
\providecommand \bibinfo  [0]{\@secondoftwo}%
\providecommand \bibfield  [0]{\@secondoftwo}%
\providecommand \translation [1]{[#1]}%
\providecommand \BibitemOpen [0]{}%
\providecommand \bibitemStop [0]{}%
\providecommand \bibitemNoStop [0]{.\EOS\space}%
\providecommand \EOS [0]{\spacefactor3000\relax}%
\providecommand \BibitemShut  [1]{\csname bibitem#1\endcsname}%
\let\auto@bib@innerbib\@empty
%</preamble>
\bibitem [{\citenamefont {Nielsen}\ and\ \citenamefont {Chuang}(2010)}]{Nielsen_Chuang_2010}%
  \BibitemOpen
  \bibfield  {author} {\bibinfo {author} {\bibfnamefont {M.~A.}\ \bibnamefont {Nielsen}}\ and\ \bibinfo {author} {\bibfnamefont {I.~L.}\ \bibnamefont {Chuang}},\ }\href@noop {} {\emph {\bibinfo {title} {Quantum Computation and Quantum Information}}}\ (\bibinfo  {publisher} {Cambridge University Press},\ \bibinfo {year} {2010})\BibitemShut {NoStop}%
\bibitem [{\citenamefont {Brandão}\ \emph {et~al.}(2015)\citenamefont {Brandão}, \citenamefont {Horodecki}, \citenamefont {Ng}, \citenamefont {Oppenheim},\ and\ \citenamefont {Wehner}}]{Brandaoetal2015}%
  \BibitemOpen
  \bibfield  {author} {\bibinfo {author} {\bibfnamefont {F.}~\bibnamefont {Brandão}}, \bibinfo {author} {\bibfnamefont {M.}~\bibnamefont {Horodecki}}, \bibinfo {author} {\bibfnamefont {N.}~\bibnamefont {Ng}}, \bibinfo {author} {\bibfnamefont {J.}~\bibnamefont {Oppenheim}},\ and\ \bibinfo {author} {\bibfnamefont {S.}~\bibnamefont {Wehner}},\ }\href {https://doi.org/10.1073/pnas.1411728112} {\bibfield  {journal} {\bibinfo  {journal} {Proceedings of the National Academy of Sciences}\ }\textbf {\bibinfo {volume} {112}},\ \bibinfo {pages} {3275} (\bibinfo {year} {2015})},\ \Eprint {https://arxiv.org/abs/https://www.pnas.org/doi/pdf/10.1073/pnas.1411728112} {https://www.pnas.org/doi/pdf/10.1073/pnas.1411728112} \BibitemShut {NoStop}%
\bibitem [{\citenamefont {Chitambar}\ and\ \citenamefont {Gour}(2019)}]{ChitambarGour2019}%
  \BibitemOpen
  \bibfield  {author} {\bibinfo {author} {\bibfnamefont {E.}~\bibnamefont {Chitambar}}\ and\ \bibinfo {author} {\bibfnamefont {G.}~\bibnamefont {Gour}},\ }\href {https://doi.org/10.1103/RevModPhys.91.025001} {\bibfield  {journal} {\bibinfo  {journal} {Rev. Mod. Phys.}\ }\textbf {\bibinfo {volume} {91}},\ \bibinfo {pages} {025001} (\bibinfo {year} {2019})}\BibitemShut {NoStop}%
\bibitem [{\citenamefont {Bennett}\ \emph {et~al.}(1993)\citenamefont {Bennett}, \citenamefont {Brassard}, \citenamefont {Cr\'epeau}, \citenamefont {Jozsa}, \citenamefont {Peres},\ and\ \citenamefont {Wootters}}]{Bennettetal1993}%
  \BibitemOpen
  \bibfield  {author} {\bibinfo {author} {\bibfnamefont {C.~H.}\ \bibnamefont {Bennett}}, \bibinfo {author} {\bibfnamefont {G.}~\bibnamefont {Brassard}}, \bibinfo {author} {\bibfnamefont {C.}~\bibnamefont {Cr\'epeau}}, \bibinfo {author} {\bibfnamefont {R.}~\bibnamefont {Jozsa}}, \bibinfo {author} {\bibfnamefont {A.}~\bibnamefont {Peres}},\ and\ \bibinfo {author} {\bibfnamefont {W.~K.}\ \bibnamefont {Wootters}},\ }\href {https://doi.org/10.1103/PhysRevLett.70.1895} {\bibfield  {journal} {\bibinfo  {journal} {Phys. Rev. Lett.}\ }\textbf {\bibinfo {volume} {70}},\ \bibinfo {pages} {1895} (\bibinfo {year} {1993})}\BibitemShut {NoStop}%
\bibitem [{\citenamefont {Popescu}(1995)}]{Popescu1995}%
  \BibitemOpen
  \bibfield  {author} {\bibinfo {author} {\bibfnamefont {S.}~\bibnamefont {Popescu}},\ }\href {https://doi.org/10.1103/PhysRevLett.74.2619} {\bibfield  {journal} {\bibinfo  {journal} {Phys. Rev. Lett.}\ }\textbf {\bibinfo {volume} {74}},\ \bibinfo {pages} {2619} (\bibinfo {year} {1995})}\BibitemShut {NoStop}%
\bibitem [{\citenamefont {Wootters}\ and\ \citenamefont {Zurek}(1982)}]{Wootters1982}%
  \BibitemOpen
  \bibfield  {author} {\bibinfo {author} {\bibfnamefont {W.~K.}\ \bibnamefont {Wootters}}\ and\ \bibinfo {author} {\bibfnamefont {W.~H.}\ \bibnamefont {Zurek}},\ }\href {https://doi.org/10.1038/299802a0} {\bibfield  {journal} {\bibinfo  {journal} {Nature}\ }\textbf {\bibinfo {volume} {299}},\ \bibinfo {pages} {802} (\bibinfo {year} {1982})}\BibitemShut {NoStop}%
\bibitem [{\citenamefont {Bu\ifmmode~\check{z}\else \v{z}\fi{}ek}\ and\ \citenamefont {Hillery}(1996)}]{BuczekHillery1996}%
  \BibitemOpen
  \bibfield  {author} {\bibinfo {author} {\bibfnamefont {V.}~\bibnamefont {Bu\ifmmode~\check{z}\else \v{z}\fi{}ek}}\ and\ \bibinfo {author} {\bibfnamefont {M.}~\bibnamefont {Hillery}},\ }\href {https://doi.org/10.1103/PhysRevA.54.1844} {\bibfield  {journal} {\bibinfo  {journal} {Phys. Rev. A}\ }\textbf {\bibinfo {volume} {54}},\ \bibinfo {pages} {1844} (\bibinfo {year} {1996})}\BibitemShut {NoStop}%
\bibitem [{\citenamefont {Coffman}\ \emph {et~al.}(2000)\citenamefont {Coffman}, \citenamefont {Kundu},\ and\ \citenamefont {Wootters}}]{Coffmanetal2000}%
  \BibitemOpen
  \bibfield  {author} {\bibinfo {author} {\bibfnamefont {V.}~\bibnamefont {Coffman}}, \bibinfo {author} {\bibfnamefont {J.}~\bibnamefont {Kundu}},\ and\ \bibinfo {author} {\bibfnamefont {W.~K.}\ \bibnamefont {Wootters}},\ }\href {https://doi.org/10.1103/PhysRevA.61.052306} {\bibfield  {journal} {\bibinfo  {journal} {Phys. Rev. A}\ }\textbf {\bibinfo {volume} {61}},\ \bibinfo {pages} {052306} (\bibinfo {year} {2000})}\BibitemShut {NoStop}%
\bibitem [{\citenamefont {Barnum}\ \emph {et~al.}(1996)\citenamefont {Barnum}, \citenamefont {Caves}, \citenamefont {Fuchs}, \citenamefont {Jozsa},\ and\ \citenamefont {Schumacher}}]{Barnumetal1996}%
  \BibitemOpen
  \bibfield  {author} {\bibinfo {author} {\bibfnamefont {H.}~\bibnamefont {Barnum}}, \bibinfo {author} {\bibfnamefont {C.~M.}\ \bibnamefont {Caves}}, \bibinfo {author} {\bibfnamefont {C.~A.}\ \bibnamefont {Fuchs}}, \bibinfo {author} {\bibfnamefont {R.}~\bibnamefont {Jozsa}},\ and\ \bibinfo {author} {\bibfnamefont {B.}~\bibnamefont {Schumacher}},\ }\href {https://doi.org/10.1103/PhysRevLett.76.2818} {\bibfield  {journal} {\bibinfo  {journal} {Phys. Rev. Lett.}\ }\textbf {\bibinfo {volume} {76}},\ \bibinfo {pages} {2818} (\bibinfo {year} {1996})}\BibitemShut {NoStop}%
\bibitem [{\citenamefont {Braunstein}\ and\ \citenamefont {Caves}(1994)}]{BraunsteinCaves1994}%
  \BibitemOpen
  \bibfield  {author} {\bibinfo {author} {\bibfnamefont {S.~L.}\ \bibnamefont {Braunstein}}\ and\ \bibinfo {author} {\bibfnamefont {C.~M.}\ \bibnamefont {Caves}},\ }\href {https://doi.org/10.1103/PhysRevLett.72.3439} {\bibfield  {journal} {\bibinfo  {journal} {Phys. Rev. Lett.}\ }\textbf {\bibinfo {volume} {72}},\ \bibinfo {pages} {3439} (\bibinfo {year} {1994})}\BibitemShut {NoStop}%
\bibitem [{\citenamefont {Giovannetti}\ \emph {et~al.}(2011)\citenamefont {Giovannetti}, \citenamefont {Lloyd},\ and\ \citenamefont {Maccone}}]{Giovannettietal2011}%
  \BibitemOpen
  \bibfield  {author} {\bibinfo {author} {\bibfnamefont {V.}~\bibnamefont {Giovannetti}}, \bibinfo {author} {\bibfnamefont {S.}~\bibnamefont {Lloyd}},\ and\ \bibinfo {author} {\bibfnamefont {L.}~\bibnamefont {Maccone}},\ }\href {https://doi.org/10.1038/nphoton.2011.35} {\bibfield  {journal} {\bibinfo  {journal} {Nature Photonics}\ }\textbf {\bibinfo {volume} {5}},\ \bibinfo {pages} {222–229} (\bibinfo {year} {2011})}\BibitemShut {NoStop}%
\bibitem [{\citenamefont {Petz}(1996)}]{Petz1996}%
  \BibitemOpen
  \bibfield  {author} {\bibinfo {author} {\bibfnamefont {D.}~\bibnamefont {Petz}},\ }\href {https://doi.org/https://doi.org/10.1016/0024-3795(94)00211-8} {\bibfield  {journal} {\bibinfo  {journal} {Linear Algebra and its Applications}\ }\textbf {\bibinfo {volume} {244}},\ \bibinfo {pages} {81} (\bibinfo {year} {1996})}\BibitemShut {NoStop}%
\bibitem [{\citenamefont {Coles}\ \emph {et~al.}(2017)\citenamefont {Coles}, \citenamefont {Berta}, \citenamefont {Tomamichel},\ and\ \citenamefont {Wehner}}]{Colesetal2017}%
  \BibitemOpen
  \bibfield  {author} {\bibinfo {author} {\bibfnamefont {P.~J.}\ \bibnamefont {Coles}}, \bibinfo {author} {\bibfnamefont {M.}~\bibnamefont {Berta}}, \bibinfo {author} {\bibfnamefont {M.}~\bibnamefont {Tomamichel}},\ and\ \bibinfo {author} {\bibfnamefont {S.}~\bibnamefont {Wehner}},\ }\href {https://doi.org/10.1103/RevModPhys.89.015002} {\bibfield  {journal} {\bibinfo  {journal} {Rev. Mod. Phys.}\ }\textbf {\bibinfo {volume} {89}},\ \bibinfo {pages} {015002} (\bibinfo {year} {2017})}\BibitemShut {NoStop}%
\bibitem [{\citenamefont {Berta}\ \emph {et~al.}(2010)\citenamefont {Berta}, \citenamefont {Christandl}, \citenamefont {Colbeck}, \citenamefont {Renes},\ and\ \citenamefont {Renner}}]{Berta2010}%
  \BibitemOpen
  \bibfield  {author} {\bibinfo {author} {\bibfnamefont {M.}~\bibnamefont {Berta}}, \bibinfo {author} {\bibfnamefont {M.}~\bibnamefont {Christandl}}, \bibinfo {author} {\bibfnamefont {R.}~\bibnamefont {Colbeck}}, \bibinfo {author} {\bibfnamefont {J.~M.}\ \bibnamefont {Renes}},\ and\ \bibinfo {author} {\bibfnamefont {R.}~\bibnamefont {Renner}},\ }\href {https://doi.org/10.1038/nphys1734} {\bibfield  {journal} {\bibinfo  {journal} {Nature Physics}\ }\textbf {\bibinfo {volume} {6}},\ \bibinfo {pages} {659–662} (\bibinfo {year} {2010})}\BibitemShut {NoStop}%
\bibitem [{\citenamefont {Horodecki}\ \emph {et~al.}(2009)\citenamefont {Horodecki}, \citenamefont {Horodecki}, \citenamefont {Horodecki},\ and\ \citenamefont {Horodecki}}]{Horodecki_2009}%
  \BibitemOpen
  \bibfield  {author} {\bibinfo {author} {\bibfnamefont {R.}~\bibnamefont {Horodecki}}, \bibinfo {author} {\bibfnamefont {P.}~\bibnamefont {Horodecki}}, \bibinfo {author} {\bibfnamefont {M.}~\bibnamefont {Horodecki}},\ and\ \bibinfo {author} {\bibfnamefont {K.}~\bibnamefont {Horodecki}},\ }\href {https://doi.org/10.1103/revmodphys.81.865} {\bibfield  {journal} {\bibinfo  {journal} {Reviews of Modern Physics}\ }\textbf {\bibinfo {volume} {81}},\ \bibinfo {pages} {865–942} (\bibinfo {year} {2009})}\BibitemShut {NoStop}%
\bibitem [{\citenamefont {Tóth}\ and\ \citenamefont {Apellaniz}(2014)}]{TothApellaniz2014}%
  \BibitemOpen
  \bibfield  {author} {\bibinfo {author} {\bibfnamefont {G.}~\bibnamefont {Tóth}}\ and\ \bibinfo {author} {\bibfnamefont {I.}~\bibnamefont {Apellaniz}},\ }\href {https://doi.org/10.1088/1751-8113/47/42/424006} {\bibfield  {journal} {\bibinfo  {journal} {Journal of Physics A: Mathematical and Theoretical}\ }\textbf {\bibinfo {volume} {47}},\ \bibinfo {pages} {424006} (\bibinfo {year} {2014})}\BibitemShut {NoStop}%
\bibitem [{\citenamefont {Petz}\ and\ \citenamefont {Ghinea}(2011)}]{Petz2011}%
  \BibitemOpen
  \bibfield  {author} {\bibinfo {author} {\bibfnamefont {D.}~\bibnamefont {Petz}}\ and\ \bibinfo {author} {\bibfnamefont {C.}~\bibnamefont {Ghinea}},\ }in\ \href {https://doi.org/10.1142/9789814338745_0015} {\emph {\bibinfo {booktitle} {Quantum Probability and Related Topics}}}\ (\bibinfo  {publisher} {World Scientific},\ \bibinfo {year} {2011})\BibitemShut {NoStop}%
\bibitem [{\citenamefont {Werner}(1998)}]{Werner1998}%
  \BibitemOpen
  \bibfield  {author} {\bibinfo {author} {\bibfnamefont {R.~F.}\ \bibnamefont {Werner}},\ }\href@noop {} {\bibfield  {journal} {\bibinfo  {journal} {Phys. Rev. A}\ }\textbf {\bibinfo {volume} {58}},\ \bibinfo {pages} {1827} (\bibinfo {year} {1998})}\BibitemShut {NoStop}%
\bibitem [{\citenamefont {Holevo}(2011)}]{Holevo2011}%
  \BibitemOpen
  \bibfield  {author} {\bibinfo {author} {\bibfnamefont {A.}~\bibnamefont {Holevo}},\ }\href@noop {} {\emph {\bibinfo {title} {Probabilistic and Statistical Aspects of Quantum Theory}}}\ (\bibinfo  {publisher} {Edizioni della Normale Pisa},\ \bibinfo {year} {2011})\BibitemShut {NoStop}%
\bibitem [{\citenamefont {Keyl}\ and\ \citenamefont {Werner}(1999)}]{Keyl1999}%
  \BibitemOpen
  \bibfield  {author} {\bibinfo {author} {\bibfnamefont {M.}~\bibnamefont {Keyl}}\ and\ \bibinfo {author} {\bibfnamefont {R.~F.}\ \bibnamefont {Werner}},\ }\href@noop {} {\bibfield  {journal} {\bibinfo  {journal} {J. Math. Phys.}\ }\textbf {\bibinfo {volume} {40}},\ \bibinfo {pages} {3283} (\bibinfo {year} {1999})}\BibitemShut {NoStop}%
\bibitem [{\citenamefont {Coecke}\ \emph {et~al.}(2016)\citenamefont {Coecke}, \citenamefont {Fritz},\ and\ \citenamefont {Spekkens}}]{Coecke_2016}%
  \BibitemOpen
  \bibfield  {author} {\bibinfo {author} {\bibfnamefont {B.}~\bibnamefont {Coecke}}, \bibinfo {author} {\bibfnamefont {T.}~\bibnamefont {Fritz}},\ and\ \bibinfo {author} {\bibfnamefont {R.~W.}\ \bibnamefont {Spekkens}},\ }\href {https://doi.org/10.1016/j.ic.2016.02.008} {\bibfield  {journal} {\bibinfo  {journal} {Information and Computation}\ }\textbf {\bibinfo {volume} {250}},\ \bibinfo {pages} {59–86} (\bibinfo {year} {2016})}\BibitemShut {NoStop}%
\bibitem [{\citenamefont {Goold}\ \emph {et~al.}(2016)\citenamefont {Goold}, \citenamefont {Huber}, \citenamefont {Riera}, \citenamefont {del Rio},\ and\ \citenamefont {Skrzypczyk}}]{Gooldetal2016}%
  \BibitemOpen
  \bibfield  {author} {\bibinfo {author} {\bibfnamefont {J.}~\bibnamefont {Goold}}, \bibinfo {author} {\bibfnamefont {M.}~\bibnamefont {Huber}}, \bibinfo {author} {\bibfnamefont {A.}~\bibnamefont {Riera}}, \bibinfo {author} {\bibfnamefont {L.}~\bibnamefont {del Rio}},\ and\ \bibinfo {author} {\bibfnamefont {P.}~\bibnamefont {Skrzypczyk}},\ }\href@noop {} {\bibfield  {journal} {\bibinfo  {journal} {J. Phys. A: Math. Theor.}\ }\textbf {\bibinfo {volume} {49}},\ \bibinfo {pages} {143001} (\bibinfo {year} {2016})}\BibitemShut {NoStop}%
\bibitem [{\citenamefont {Arnold}(2013)}]{Arnold2013}%
  \BibitemOpen
  \bibfield  {author} {\bibinfo {author} {\bibfnamefont {V.~I.}\ \bibnamefont {Arnold}},\ }\href@noop {} {\emph {\bibinfo {title} {Mathematical Methods of Classical Mechanics}}}\ (\bibinfo  {publisher} {Springer},\ \bibinfo {year} {2013})\BibitemShut {NoStop}%
\bibitem [{\citenamefont {Ruskai}\ \emph {et~al.}(2002)\citenamefont {Ruskai}, \citenamefont {Szarek},\ and\ \citenamefont {Werner}}]{Ruskaietal2002}%
  \BibitemOpen
  \bibfield  {author} {\bibinfo {author} {\bibfnamefont {M.~B.}\ \bibnamefont {Ruskai}}, \bibinfo {author} {\bibfnamefont {S.}~\bibnamefont {Szarek}},\ and\ \bibinfo {author} {\bibfnamefont {R.}~\bibnamefont {Werner}},\ }\href@noop {} {\bibfield  {journal} {\bibinfo  {journal} {Linear Algebra and its Applications}\ }\textbf {\bibinfo {volume} {347}},\ \bibinfo {pages} {159} (\bibinfo {year} {2002})}\BibitemShut {NoStop}%
\bibitem [{\citenamefont {Wilde}(2013)}]{Wilde_2013}%
  \BibitemOpen
  \bibfield  {author} {\bibinfo {author} {\bibfnamefont {M.~M.}\ \bibnamefont {Wilde}},\ }\href@noop {} {\emph {\bibinfo {title} {Quantum Information Theory}}}\ (\bibinfo  {publisher} {Cambridge University Press},\ \bibinfo {year} {2013})\BibitemShut {NoStop}%
\bibitem [{\citenamefont {Zurek}(2003)}]{Zurek2003}%
  \BibitemOpen
  \bibfield  {author} {\bibinfo {author} {\bibfnamefont {W.~H.}\ \bibnamefont {Zurek}},\ }\href@noop {} {\bibfield  {journal} {\bibinfo  {journal} {Rev. Mod. Phys.}\ }\textbf {\bibinfo {volume} {75}},\ \bibinfo {pages} {715} (\bibinfo {year} {2003})}\BibitemShut {NoStop}%
\bibitem [{\citenamefont {Vedral}(2002)}]{Vedral2002}%
  \BibitemOpen
  \bibfield  {author} {\bibinfo {author} {\bibfnamefont {V.}~\bibnamefont {Vedral}},\ }\href@noop {} {\bibfield  {journal} {\bibinfo  {journal} {Rev. Mod. Phys.}\ }\textbf {\bibinfo {volume} {74}},\ \bibinfo {pages} {197} (\bibinfo {year} {2002})}\BibitemShut {NoStop}%
\bibitem [{\citenamefont {Brand\~ao}\ and\ \citenamefont {Gour}(2015)}]{BrandaoGour2015}%
  \BibitemOpen
  \bibfield  {author} {\bibinfo {author} {\bibfnamefont {F.~G. S.~L.}\ \bibnamefont {Brand\~ao}}\ and\ \bibinfo {author} {\bibfnamefont {G.}~\bibnamefont {Gour}},\ }\href {https://doi.org/10.1103/PhysRevLett.115.070503} {\bibfield  {journal} {\bibinfo  {journal} {Phys. Rev. Lett.}\ }\textbf {\bibinfo {volume} {115}},\ \bibinfo {pages} {070503} (\bibinfo {year} {2015})}\BibitemShut {NoStop}%
\bibitem [{\citenamefont {Brandão}\ and\ \citenamefont {Gour}(2011)}]{Brandaoetal2011}%
  \BibitemOpen
  \bibfield  {author} {\bibinfo {author} {\bibfnamefont {F.~G.}\ \bibnamefont {Brandão}}\ and\ \bibinfo {author} {\bibfnamefont {G.}~\bibnamefont {Gour}},\ }\href@noop {} {\bibfield  {journal} {\bibinfo  {journal} {Phys. Rev. Lett.}\ }\textbf {\bibinfo {volume} {106}},\ \bibinfo {pages} {050501} (\bibinfo {year} {2011})}\BibitemShut {NoStop}%
\bibitem [{\citenamefont {Meter}(2014)}]{VanMeter2014}%
  \BibitemOpen
  \bibfield  {author} {\bibinfo {author} {\bibfnamefont {R.~V.}\ \bibnamefont {Meter}},\ }\href@noop {} {\emph {\bibinfo {title} {Quantum Networking}}}\ (\bibinfo  {publisher} {John Wiley \& Sons},\ \bibinfo {address} {Chichester, UK},\ \bibinfo {year} {2014})\BibitemShut {NoStop}%
\bibitem [{\citenamefont {Bouwmeester}\ \emph {et~al.}(1997)\citenamefont {Bouwmeester}, \citenamefont {Pan}, \citenamefont {Mattle}, \citenamefont {Eibl}, \citenamefont {Weinfurter},\ and\ \citenamefont {Zeilinger}}]{Bouwmeesteretal1997}%
  \BibitemOpen
  \bibfield  {author} {\bibinfo {author} {\bibfnamefont {D.}~\bibnamefont {Bouwmeester}}, \bibinfo {author} {\bibfnamefont {J.-W.}\ \bibnamefont {Pan}}, \bibinfo {author} {\bibfnamefont {K.}~\bibnamefont {Mattle}}, \bibinfo {author} {\bibfnamefont {M.}~\bibnamefont {Eibl}}, \bibinfo {author} {\bibfnamefont {H.}~\bibnamefont {Weinfurter}},\ and\ \bibinfo {author} {\bibfnamefont {A.}~\bibnamefont {Zeilinger}},\ }\href {https://doi.org/10.1038/37539} {\bibfield  {journal} {\bibinfo  {journal} {Nature}\ }\textbf {\bibinfo {volume} {390}},\ \bibinfo {pages} {575–579} (\bibinfo {year} {1997})}\BibitemShut {NoStop}%
\bibitem [{\citenamefont {Steffen}\ \emph {et~al.}(2013)\citenamefont {Steffen}, \citenamefont {Salathe}, \citenamefont {Oppliger}, \citenamefont {Kurpiers}, \citenamefont {Baur}, \citenamefont {Lang}, \citenamefont {Eichler}, \citenamefont {Puebla-Hellmann}, \citenamefont {Fedorov},\ and\ \citenamefont {Wallraff}}]{Steffenetal2013}%
  \BibitemOpen
  \bibfield  {author} {\bibinfo {author} {\bibfnamefont {L.}~\bibnamefont {Steffen}}, \bibinfo {author} {\bibfnamefont {Y.}~\bibnamefont {Salathe}}, \bibinfo {author} {\bibfnamefont {M.}~\bibnamefont {Oppliger}}, \bibinfo {author} {\bibfnamefont {P.}~\bibnamefont {Kurpiers}}, \bibinfo {author} {\bibfnamefont {M.}~\bibnamefont {Baur}}, \bibinfo {author} {\bibfnamefont {C.}~\bibnamefont {Lang}}, \bibinfo {author} {\bibfnamefont {C.}~\bibnamefont {Eichler}}, \bibinfo {author} {\bibfnamefont {G.}~\bibnamefont {Puebla-Hellmann}}, \bibinfo {author} {\bibfnamefont {A.}~\bibnamefont {Fedorov}},\ and\ \bibinfo {author} {\bibfnamefont {A.}~\bibnamefont {Wallraff}},\ }\href {https://doi.org/10.1038/nature12422} {\bibfield  {journal} {\bibinfo  {journal} {Nature}\ }\textbf {\bibinfo {volume} {500}},\ \bibinfo {pages} {319–322} (\bibinfo {year} {2013})}\BibitemShut {NoStop}%
\bibitem [{\citenamefont {Scarani}\ \emph {et~al.}(2005)\citenamefont {Scarani}, \citenamefont {Iblisdir}, \citenamefont {Gisin},\ and\ \citenamefont {Ac\'{\i}n}}]{Scaranietal2005}%
  \BibitemOpen
  \bibfield  {author} {\bibinfo {author} {\bibfnamefont {V.}~\bibnamefont {Scarani}}, \bibinfo {author} {\bibfnamefont {S.}~\bibnamefont {Iblisdir}}, \bibinfo {author} {\bibfnamefont {N.}~\bibnamefont {Gisin}},\ and\ \bibinfo {author} {\bibfnamefont {A.}~\bibnamefont {Ac\'{\i}n}},\ }\href {https://doi.org/10.1103/RevModPhys.77.1225} {\bibfield  {journal} {\bibinfo  {journal} {Rev. Mod. Phys.}\ }\textbf {\bibinfo {volume} {77}},\ \bibinfo {pages} {1225} (\bibinfo {year} {2005})}\BibitemShut {NoStop}%
\bibitem [{\citenamefont {Lamas-Linares}\ \emph {et~al.}(2002)\citenamefont {Lamas-Linares}, \citenamefont {Simon}, \citenamefont {Howell},\ and\ \citenamefont {Bouwmeester}}]{LamasLinaresetal002}%
  \BibitemOpen
  \bibfield  {author} {\bibinfo {author} {\bibfnamefont {A.}~\bibnamefont {Lamas-Linares}}, \bibinfo {author} {\bibfnamefont {C.}~\bibnamefont {Simon}}, \bibinfo {author} {\bibfnamefont {J.~C.}\ \bibnamefont {Howell}},\ and\ \bibinfo {author} {\bibfnamefont {D.}~\bibnamefont {Bouwmeester}},\ }\href {https://doi.org/10.1126/science.1068972} {\bibfield  {journal} {\bibinfo  {journal} {Science}\ }\textbf {\bibinfo {volume} {296}},\ \bibinfo {pages} {712–714} (\bibinfo {year} {2002})}\BibitemShut {NoStop}%
\bibitem [{\citenamefont {Du}\ \emph {et~al.}(2005)\citenamefont {Du}, \citenamefont {Durt}, \citenamefont {Zou}, \citenamefont {Li}, \citenamefont {Kwek}, \citenamefont {Lai}, \citenamefont {Oh},\ and\ \citenamefont {Ekert}}]{Duetal2005}%
  \BibitemOpen
  \bibfield  {author} {\bibinfo {author} {\bibfnamefont {J.}~\bibnamefont {Du}}, \bibinfo {author} {\bibfnamefont {T.}~\bibnamefont {Durt}}, \bibinfo {author} {\bibfnamefont {P.}~\bibnamefont {Zou}}, \bibinfo {author} {\bibfnamefont {H.}~\bibnamefont {Li}}, \bibinfo {author} {\bibfnamefont {L.~C.}\ \bibnamefont {Kwek}}, \bibinfo {author} {\bibfnamefont {C.~H.}\ \bibnamefont {Lai}}, \bibinfo {author} {\bibfnamefont {C.~H.}\ \bibnamefont {Oh}},\ and\ \bibinfo {author} {\bibfnamefont {A.}~\bibnamefont {Ekert}},\ }\href {https://doi.org/10.1103/PhysRevLett.94.040505} {\bibfield  {journal} {\bibinfo  {journal} {Phys. Rev. Lett.}\ }\textbf {\bibinfo {volume} {94}},\ \bibinfo {pages} {040505} (\bibinfo {year} {2005})}\BibitemShut {NoStop}%
\bibitem [{\citenamefont {Pezz\'e}\ and\ \citenamefont {Smerzi}(2009)}]{PezzeSmerzi2009}%
  \BibitemOpen
  \bibfield  {author} {\bibinfo {author} {\bibfnamefont {L.}~\bibnamefont {Pezz\'e}}\ and\ \bibinfo {author} {\bibfnamefont {A.}~\bibnamefont {Smerzi}},\ }\href {https://doi.org/10.1103/PhysRevLett.102.100401} {\bibfield  {journal} {\bibinfo  {journal} {Phys. Rev. Lett.}\ }\textbf {\bibinfo {volume} {102}},\ \bibinfo {pages} {100401} (\bibinfo {year} {2009})}\BibitemShut {NoStop}%
\bibitem [{\citenamefont {Greenberger}\ \emph {et~al.}(1990)\citenamefont {Greenberger}, \citenamefont {Horne}, \citenamefont {Shimony},\ and\ \citenamefont {Zeilinger}}]{Greenbergeretal1990}%
  \BibitemOpen
  \bibfield  {author} {\bibinfo {author} {\bibfnamefont {D.~M.}\ \bibnamefont {Greenberger}}, \bibinfo {author} {\bibfnamefont {M.~A.}\ \bibnamefont {Horne}}, \bibinfo {author} {\bibfnamefont {A.}~\bibnamefont {Shimony}},\ and\ \bibinfo {author} {\bibfnamefont {A.}~\bibnamefont {Zeilinger}},\ }\href {https://doi.org/10.1119/1.16243} {\bibfield  {journal} {\bibinfo  {journal} {American Journal of Physics}\ }\textbf {\bibinfo {volume} {58}},\ \bibinfo {pages} {1131} (\bibinfo {year} {1990})},\ \Eprint {https://arxiv.org/abs/https://pubs.aip.org/aapt/ajp/article-pdf/58/12/1131/11479397/1131\_1\_online.pdf} {https://pubs.aip.org/aapt/ajp/article-pdf/58/12/1131/11479397/1131\_1\_online.pdf} \BibitemShut {NoStop}%
\bibitem [{\citenamefont {Dür}\ \emph {et~al.}(2000)\citenamefont {Dür}, \citenamefont {Vidal},\ and\ \citenamefont {Cirac}}]{Duretal2000}%
  \BibitemOpen
  \bibfield  {author} {\bibinfo {author} {\bibfnamefont {W.}~\bibnamefont {Dür}}, \bibinfo {author} {\bibfnamefont {G.}~\bibnamefont {Vidal}},\ and\ \bibinfo {author} {\bibfnamefont {J.~I.}\ \bibnamefont {Cirac}},\ }\bibfield  {journal} {\bibinfo  {journal} {Physical Review A}\ }\textbf {\bibinfo {volume} {62}},\ \href {https://doi.org/10.1103/physreva.62.062314} {10.1103/physreva.62.062314} (\bibinfo {year} {2000})\BibitemShut {NoStop}%
\bibitem [{\citenamefont {Werner}(1989)}]{Werner1989}%
  \BibitemOpen
  \bibfield  {author} {\bibinfo {author} {\bibfnamefont {R.~F.}\ \bibnamefont {Werner}},\ }\href {https://doi.org/10.1103/PhysRevA.40.4277} {\bibfield  {journal} {\bibinfo  {journal} {Phys. Rev. A}\ }\textbf {\bibinfo {volume} {40}},\ \bibinfo {pages} {4277} (\bibinfo {year} {1989})}\BibitemShut {NoStop}%
\bibitem [{\citenamefont {Koashi}\ and\ \citenamefont {Winter}(2004)}]{KoashiWinter2004}%
  \BibitemOpen
  \bibfield  {author} {\bibinfo {author} {\bibfnamefont {M.}~\bibnamefont {Koashi}}\ and\ \bibinfo {author} {\bibfnamefont {A.}~\bibnamefont {Winter}},\ }\bibfield  {journal} {\bibinfo  {journal} {Physical Review A}\ }\textbf {\bibinfo {volume} {69}},\ \href {https://doi.org/10.1103/physreva.69.022309} {10.1103/physreva.69.022309} (\bibinfo {year} {2004})\BibitemShut {NoStop}%
\bibitem [{\citenamefont {Baumgratz}\ \emph {et~al.}(2014)\citenamefont {Baumgratz}, \citenamefont {Cramer},\ and\ \citenamefont {Plenio}}]{Baumgratz2014}%
  \BibitemOpen
  \bibfield  {author} {\bibinfo {author} {\bibfnamefont {T.}~\bibnamefont {Baumgratz}}, \bibinfo {author} {\bibfnamefont {M.}~\bibnamefont {Cramer}},\ and\ \bibinfo {author} {\bibfnamefont {M.~B.}\ \bibnamefont {Plenio}},\ }\href@noop {} {\bibfield  {journal} {\bibinfo  {journal} {Phys. Rev. Lett.}\ }\textbf {\bibinfo {volume} {113}},\ \bibinfo {pages} {140401} (\bibinfo {year} {2014})}\BibitemShut {NoStop}%
\bibitem [{\citenamefont {Streltsov}\ \emph {et~al.}(2017)\citenamefont {Streltsov}, \citenamefont {Adesso},\ and\ \citenamefont {Plenio}}]{Streltsov2017}%
  \BibitemOpen
  \bibfield  {author} {\bibinfo {author} {\bibfnamefont {A.}~\bibnamefont {Streltsov}}, \bibinfo {author} {\bibfnamefont {G.}~\bibnamefont {Adesso}},\ and\ \bibinfo {author} {\bibfnamefont {M.~B.}\ \bibnamefont {Plenio}},\ }\href@noop {} {\bibfield  {journal} {\bibinfo  {journal} {Rev. Mod. Phys.}\ }\textbf {\bibinfo {volume} {89}},\ \bibinfo {pages} {041003} (\bibinfo {year} {2017})}\BibitemShut {NoStop}%
\bibitem [{\citenamefont {Marvian}\ and\ \citenamefont {Spekkens}(2014)}]{MarvianSpekkens2014}%
  \BibitemOpen
  \bibfield  {author} {\bibinfo {author} {\bibfnamefont {I.}~\bibnamefont {Marvian}}\ and\ \bibinfo {author} {\bibfnamefont {R.~W.}\ \bibnamefont {Spekkens}},\ }\href@noop {} {\bibfield  {journal} {\bibinfo  {journal} {Nat. Commun.}\ }\textbf {\bibinfo {volume} {5}},\ \bibinfo {pages} {3821} (\bibinfo {year} {2014})}\BibitemShut {NoStop}%
\bibitem [{\citenamefont {Lostaglio}\ \emph {et~al.}(2015)\citenamefont {Lostaglio}, \citenamefont {Jennings},\ and\ \citenamefont {Rudolph}}]{Lostaglioetal2015}%
  \BibitemOpen
  \bibfield  {author} {\bibinfo {author} {\bibfnamefont {M.}~\bibnamefont {Lostaglio}}, \bibinfo {author} {\bibfnamefont {D.}~\bibnamefont {Jennings}},\ and\ \bibinfo {author} {\bibfnamefont {T.}~\bibnamefont {Rudolph}},\ }\bibfield  {journal} {\bibinfo  {journal} {Nature Communications}\ }\textbf {\bibinfo {volume} {6}},\ \href {https://doi.org/10.1038/ncomms7383} {10.1038/ncomms7383} (\bibinfo {year} {2015})\BibitemShut {NoStop}%
\bibitem [{\citenamefont {Wehner}\ \emph {et~al.}(2018)\citenamefont {Wehner}, \citenamefont {Elkouss},\ and\ \citenamefont {Hanson}}]{Wehneretal2018}%
  \BibitemOpen
  \bibfield  {author} {\bibinfo {author} {\bibfnamefont {S.}~\bibnamefont {Wehner}}, \bibinfo {author} {\bibfnamefont {D.}~\bibnamefont {Elkouss}},\ and\ \bibinfo {author} {\bibfnamefont {R.}~\bibnamefont {Hanson}},\ }\href {https://doi.org/10.1126/science.aam9288} {\bibfield  {journal} {\bibinfo  {journal} {Science}\ }\textbf {\bibinfo {volume} {362}},\ \bibinfo {pages} {eaam9288} (\bibinfo {year} {2018})},\ \Eprint {https://arxiv.org/abs/https://www.science.org/doi/pdf/10.1126/science.aam9288} {https://www.science.org/doi/pdf/10.1126/science.aam9288} \BibitemShut {NoStop}%
\bibitem [{\citenamefont {Bharti}\ \emph {et~al.}(2022)\citenamefont {Bharti}, \citenamefont {Cervera-Lierta}, \citenamefont {Kyaw}, \citenamefont {Haug}, \citenamefont {Alperin-Lea}, \citenamefont {Anand}, \citenamefont {Degroote}, \citenamefont {Heimonen}, \citenamefont {Kottmann}, \citenamefont {Menke}, \citenamefont {Mok}, \citenamefont {Sim}, \citenamefont {Kwek},\ and\ \citenamefont {Aspuru-Guzik}}]{Bhartietal2022}%
  \BibitemOpen
  \bibfield  {author} {\bibinfo {author} {\bibfnamefont {K.}~\bibnamefont {Bharti}}, \bibinfo {author} {\bibfnamefont {A.}~\bibnamefont {Cervera-Lierta}}, \bibinfo {author} {\bibfnamefont {T.~H.}\ \bibnamefont {Kyaw}}, \bibinfo {author} {\bibfnamefont {T.}~\bibnamefont {Haug}}, \bibinfo {author} {\bibfnamefont {S.}~\bibnamefont {Alperin-Lea}}, \bibinfo {author} {\bibfnamefont {A.}~\bibnamefont {Anand}}, \bibinfo {author} {\bibfnamefont {M.}~\bibnamefont {Degroote}}, \bibinfo {author} {\bibfnamefont {H.}~\bibnamefont {Heimonen}}, \bibinfo {author} {\bibfnamefont {J.~S.}\ \bibnamefont {Kottmann}}, \bibinfo {author} {\bibfnamefont {T.}~\bibnamefont {Menke}}, \bibinfo {author} {\bibfnamefont {W.-K.}\ \bibnamefont {Mok}}, \bibinfo {author} {\bibfnamefont {S.}~\bibnamefont {Sim}}, \bibinfo {author} {\bibfnamefont {L.-C.}\ \bibnamefont {Kwek}},\ and\ \bibinfo {author} {\bibfnamefont {A.}~\bibnamefont {Aspuru-Guzik}},\ }\bibfield  {journal} {\bibinfo  {journal} {Reviews of Modern Physics}\ }\textbf {\bibinfo {volume}
  {94}},\ \href {https://doi.org/10.1103/revmodphys.94.015004} {10.1103/revmodphys.94.015004} (\bibinfo {year} {2022})\BibitemShut {NoStop}%
\bibitem [{\citenamefont {Benenti}\ \emph {et~al.}(2018)\citenamefont {Benenti}, \citenamefont {Casati}, \citenamefont {Rossini},\ and\ \citenamefont {Strini}}]{Benentietal2018}%
  \BibitemOpen
  \bibfield  {author} {\bibinfo {author} {\bibfnamefont {G.}~\bibnamefont {Benenti}}, \bibinfo {author} {\bibfnamefont {G.}~\bibnamefont {Casati}}, \bibinfo {author} {\bibfnamefont {D.}~\bibnamefont {Rossini}},\ and\ \bibinfo {author} {\bibfnamefont {G.}~\bibnamefont {Strini}},\ }\href {https://doi.org/10.1142/10909} {\emph {\bibinfo {title} {Principles of Quantum Computation and Information}}}\ (\bibinfo  {publisher} {WORLD SCIENTIFIC},\ \bibinfo {year} {2018})\ \Eprint {https://arxiv.org/abs/https://www.worldscientific.com/doi/pdf/10.1142/10909} {https://www.worldscientific.com/doi/pdf/10.1142/10909} \BibitemShut {NoStop}%
\bibitem [{\citenamefont {Horodecki}\ \emph {et~al.}(1999)\citenamefont {Horodecki}, \citenamefont {Horodecki},\ and\ \citenamefont {Horodecki}}]{Horodecki_1999}%
  \BibitemOpen
  \bibfield  {author} {\bibinfo {author} {\bibfnamefont {M.}~\bibnamefont {Horodecki}}, \bibinfo {author} {\bibfnamefont {P.}~\bibnamefont {Horodecki}},\ and\ \bibinfo {author} {\bibfnamefont {R.}~\bibnamefont {Horodecki}},\ }\href {https://doi.org/10.1103/PhysRevA.60.1888} {\bibfield  {journal} {\bibinfo  {journal} {Phys. Rev. A}\ }\textbf {\bibinfo {volume} {60}},\ \bibinfo {pages} {1888} (\bibinfo {year} {1999})}\BibitemShut {NoStop}%
\bibitem [{\citenamefont {Bru\ss{}}\ \emph {et~al.}(1998)\citenamefont {Bru\ss{}}, \citenamefont {DiVincenzo}, \citenamefont {Ekert}, \citenamefont {Fuchs}, \citenamefont {Macchiavello},\ and\ \citenamefont {Smolin}}]{Bruss_1998}%
  \BibitemOpen
  \bibfield  {author} {\bibinfo {author} {\bibfnamefont {D.}~\bibnamefont {Bru\ss{}}}, \bibinfo {author} {\bibfnamefont {D.~P.}\ \bibnamefont {DiVincenzo}}, \bibinfo {author} {\bibfnamefont {A.}~\bibnamefont {Ekert}}, \bibinfo {author} {\bibfnamefont {C.~A.}\ \bibnamefont {Fuchs}}, \bibinfo {author} {\bibfnamefont {C.}~\bibnamefont {Macchiavello}},\ and\ \bibinfo {author} {\bibfnamefont {J.~A.}\ \bibnamefont {Smolin}},\ }\href {https://doi.org/10.1103/PhysRevA.57.2368} {\bibfield  {journal} {\bibinfo  {journal} {Phys. Rev. A}\ }\textbf {\bibinfo {volume} {57}},\ \bibinfo {pages} {2368} (\bibinfo {year} {1998})}\BibitemShut {NoStop}%
\end{thebibliography}%

\end{document}